\documentclass[pdflatex,sn-mathphys-num]{sn-jnl}

\usepackage{graphicx}%
\usepackage{amsmath,amssymb,amsfonts}%
\usepackage{amsthm}%
\usepackage{mathrsfs}%
\usepackage[title]{appendix}%
\usepackage{xcolor}%
\usepackage{booktabs}%

\usepackage{xspace}
\usepackage{bm}
\usepackage{todonotes}
\usepackage[font=small,labelfont=small]{subcaption}
\usepackage[font=small,labelfont=bf]{caption}
\usepackage[misc,geometry]{ifsym}
\usepackage{thmtools}
\usepackage[capitalise]{cleveref}
\usepackage{enumerate}

\theoremstyle{thmstyleone}%
\newtheorem{theorem}{Theorem}%
\newtheorem{lemma}[theorem]{Lemma}%
\newtheorem{corollary}[theorem]{Corollary}%
\newtheorem{property}[theorem]{Property}

\graphicspath{{figures/}}

\definecolor{lightcyan}{rgb}{0.88,1,1}
\definecolor{antiquewhite}{rgb}{0.98, 0.92, 0.84}
\definecolor{lightgreen}{rgb}{0.6, 0.98, 0.6}
\definecolor{defblue}{rgb}{0.121,0.47,0.705}

\DeclareTextFontCommand{\emph}{\color{defblue}\em}

\DeclareMathOperator\ord{ord}
\DeclareMathOperator\MMST{MMST}

\newcommand{\D}{\ensuremath{\mathcal D}\xspace}
\newcommand{\W}{\ensuremath{\mathcal W}\xspace}
\newcommand{\LL}{\ensuremath{\mathcal L}\xspace}
\newcommand{\Wutil}{\ensuremath{\mathcal W}\xspace}

\begin{document}

\title[Minimum Monotone Spanning Trees]{Minimum Monotone Spanning
  Trees\footnote{An extended abstract has been presented in:
    Proc.\ 50th International Conference on Current Trends
    in Theory and Practice of Computer Science (SOFSEM
    2025)~\cite{ddkssw-mmst-sofsem25}.}}

\author[1]{\fnm{Emilio} \sur{Di Giacomo}
  \orcid{https://orcid.org/0000-0002-9794-1928}}

\author[1]{\fnm{Walter} \sur{Didimo} 
  \orcid{https://orcid.org/0000-0002-4379-6059}}

\author*[2]{\fnm{Eleni} \sur{Katsanou}
  \orcid{https://orcid.org/0000-0002-1001-1411}}
\email{ekatsanou@mail.ntua.gr}

\author[3]{\fnm{Lena} \sur{Schlipf}
  \orcid{https://orcid.org/0000-0001-7043-1867}}

\author[2]{\fnm{Antonios} \sur{Symvonis}
  \orcid{https://orcid.org/0000-0002-0280-741X}}

\author[4]{\fnm{Alexander} \sur{Wolff}
  \orcid{https://orcid.org/0000-0001-5872-718X}}

\affil[1]{\orgname{Universit\`a degli Studi di Perugia}, \orgaddress{\city{Perugia}, \country{Italy}}}

\affil*[2]{\orgname{National Technical University of Athens}, \orgaddress{\city{Athens}, \country{Greece}}}

\affil[3]{\orgname{Universit\"at T\"ubingen}, \orgaddress{\city{T\"ubingen}, \country{Germany}}}

\affil[4]{\orgname{Universit\"at W\"urzburg}, \orgaddress{\city{W\"urzburg}, \country{Germany}}}
    
\abstract{Given a finite set $S$ of points in the plane and a finite
  set $\mathcal{D}$ of directions, a geometric spanning tree~$T$
  of~$S$ is \emph{$\mathcal{D}$-monotone} if every path in $T$ is
  monotone with respect to some direction in $\mathcal{D}$.  We study
  the problem of computing, for a given point set $S$ and a given set
  $\mathcal{D}$ of directions, a minimum-length ${\cal D}$-monotone
  spanning tree of~$S$.  We present a quadratic-time algorithm for two
  directions.  More generally, we show that the problem belongs to the
  complexity class XP when parameterized by the number of directions.
  We further study, for a given positive integer $k$ and point
  set~$S$, the problem of finding a minimum-length
  $\mathcal{D}$-monotone spanning tree of $S$ over all possible
  sets~$\mathcal{D}$ of $k$ directions.  We prove that this problem,
  too, is in XP when parameterized by~$k$, and present two algorithms
  that run in $O(n^2 \log n)$ and $O(n^6)$ time for $k=1$ and $k=2$,
  respectively, where $n$ is the number of points in~$S$.

  Finally, in contrast to the classical Euclidean minimum spanning
  tree of a set of points, whose vertex degree is bounded by six, we
  show that for every even integer~$k$, there exists a point set~$S_k$
  and a set $\mathcal{D}_k$ of $k$ directions such that any
  minimum-length $\mathcal{D}_k$-monotone spanning tree of $S_k$ has
  maximum vertex degree~$2k$.}


\maketitle

\section{Introduction}
\label{se:introduction}
	
We study a problem that combines the notion of minimum
spanning tree of a set of points in the plane with the notion of
monotone drawings of graphs.

The problem of computing a (Euclidean) \emph{minimum spanning tree
  (MST)} of a set of points in the plane is a well-established topic
with a long history in computational geometry~\cite{GrahamH85}.  An
MST of a finite set~$S$ of points is a geometric tree~$T$ such that:
$(i)$~$T$ \emph{spans} $S$, i.e., the vertices of $T$ are the points
of $S$, and $(ii)$~$T$ has minimum length subject to property~$(i)$,
where the length of $T$ is the sum of the lengths of its edges and the
length of an edge is the Euclidean distance of its endpoints.
Equivalently, the MST is the minimum spanning tree of the complete
graph on $S$ where the weight of each edge is the Euclidean distance
of its incident vertices.  It is known that an MST is a subgraph of a
Delaunay triangulation~\cite{shamosHoey75} (see
\cref{fi:monotone-example1,fi:monotone-example2}).  Given a set $S$ of
$n$ points, its Delaunay triangulation has at most $3n-6$ edges, hence
an MST of~$S$ can be computed in $O(n\log n)$ time (in the real RAM
model of computation) via standard MST algorithms.
Refer to the survey of Eppstein~\cite{eppstein2000} for additional details and references on MSTs.
The problem has also been studied in
higher dimensions~\cite{agarwalESW91,NarasimhanZZ2000,yao82}.

\emph{Monotone drawings of graphs} have been introduced by Angelini,
Colasante, {Di Battista}, Frati, and Patrignani~\cite{AngeliniCBFP12}
and have received considerable attention in recent years. They are
related to other types of drawings of graphs, such as
angle-monotone~\cite{bakhshesh21,bakhshesh22,bonichon16,dehkordi14,lubiw18},
upward~\cite{DBLP:reference/algo/Didimo16,gt-upt-95},
greedy~\cite{DBLP:journals/tcs/AngeliniBDGKMPS19,angelini10,dehkordi14,Dhandapani2010,Papadimitriou2005,Rao2003},
self-approaching~\cite{alamdari13,bakhshesh19,nollenburg16}, and
increasing-chord drawings~\cite{bahoo17,dehkordi14,mastakas15,nollenburg16}. Computing monotone drawings is also
related to the geometric problem of finding monotone trajectories
between two given points in the plane avoiding convex
obstacles~\cite{DBLP:conf/compgeom/ArkinCM89}.
A plane path is \emph{monotone with respect to a
  direction~$d$} if the order of its vertices along the path coincides
with the order of their projections on a line parallel to~$d$.  Any
monotone path is necessarily crossing-free~\cite{AngeliniCBFP12}.  A
straight-line drawing of a graph~$G$ in the plane is \emph{monotone}
if there exists a monotone path (with respect to some direction)
between any two vertices of~$G$; the direction of monotonicity may be
different for each path.  If the directions of monotonicity for the
paths are restricted to a set \D of directions, then the
drawing is \emph{\D-monotone}. Results about monotone
drawings include algorithms for different graph classes
\cite{angelini17,AngeliniCBFP12,angeliniDKMRSW15,felsnerIKKMS16}
and the study of the area requirement of such drawings (see
\cite{hehe17,kindermannSSW14,oikonomouSym18} for
monotone drawings of trees and
\cite{hehe15b,hossainR15,oikonomouSymv17} for different classes of
planar graphs).

\begin{figure}[hb]
  \begin{subfigure}{.32\linewidth}
    \centering
    \includegraphics[page=1]{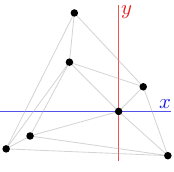}
    \subcaption{}
    \label{fi:monotone-example1}
  \end{subfigure}
  \hfill
  \begin{subfigure}{.32\linewidth}
    \centering
    \includegraphics[page=2]{monotone_tree_example}
    \subcaption{}
    \label{fi:monotone-example2}
  \end{subfigure}
  \hfill
  \begin{subfigure}{.32\linewidth}
    \centering
    \includegraphics[page=3]{monotone_tree_example}
    \subcaption{}
    \label{fi:monotone-example3}
  \end{subfigure}
  \caption{(a) A point set $S$ with its Delaunay triangulation,
    (b)~a (Euclidean) minimum spanning tree of~$S$, (c)~a minimum
    $\D$-monotone spanning tree of~$S$ w.r.t.\
    $\{\binom{1}{0},\binom{0}{1}\}$.  The $v_2$--$v_3$ path is
    $x$-monotone; the $v_1$--$v_2$ and $v_1$--$v_3$ paths are $y$-monotone.}
  \label{fi:monotone-example}
\end{figure}

\subparagraph*{Our setting.}  In this paper, we study a natural setting that combines the benefits of spanning trees of minimum length with the benefits of monotone drawings. Namely, given a set $S$ of $n$
points in the plane and a prescribed set~\D of directions, we study the
 \emph{$\MMST(S,\D)$} problem of computing
a \D-monotone spanning tree of~$S$ of minimum length (see
\cref{fi:monotone-example3}).  We call such a tree a \emph{minimum
$\D$-monotone spanning tree}.  For a point set $S$ and an integer
$k \geq 1$, we also address the \emph{$\MMST(S,k)$} problem of computing a
minimum \emph{$k$-directional monotone spanning tree} of~$S$, i.e., a
\D-monotone spanning tree of minimum length among all possible
sets~\D of $k$ directions.  In this variant, the choice of the
directions of monotonicity adjusts to the given point set, which can
lead to shorter monotone spanning trees.
	We remark that this is not the first attempt to couple the MST problem
	with an additional property.  For example, the \emph{Euclidean degree-$\Delta$ MST} asks for an MST whose maximum degree is bounded by a given integer
	$\Delta$~\cite{franckeHoffmann09,PapadimitriouVazirani84}.
	The Euclidean degree-$2$ problem coincides with the Euclidean Traveling Salesperson Path problem, thus is NP-hard. It has been also proven that the problem is NP-hard for $\Delta=3$ \cite{PapadimitriouVazirani84} and $\Delta=4$ \cite{franckeHoffmann09}. For any $\Delta\ge 5$~the problem asks for the unconstrained MST of the point set \cite{PapadimitriouVazirani84}. 
	Seo, Lee, and Lin~\cite{SeoLL09} studied the problem of computing a
	MST of smallest diameter or smallest radius.
	Finding the $k$ smallest spanning trees \cite{k-mst-eppstein,k-mst-Frederickson,k-mst-gabow} or the dynamic MSTs \cite{dmst-chin,dmst-spira} are further problems related to spanning trees.

	Particularly relevant to our study is the \emph{Rooted Monotone MST problem} introduced by Mastakas and Symvonis~\cite{mastakasSymv17}.
	In that problem, given a set $S$ of $n$ points with a designated root
	$r \in S$, the task is to compute an MST such that the path from $r$
	to any other point of $S$ is monotone.  In contrast to the $\MMST$
	problem, which requires satisfying the monotonicity of all
	$n \choose 2$ paths between the points of $S$, in the rooted version
	of the problem, monotonicity is required only for $n-1$ paths.
	Mastakas and Symvonis~\cite{mastakasSymv17} showed how to compute a
	rooted MST that is monotone with respect to one and two orthogonal
	directions.  Mastakas~\cite{mastakas18} 
	extended their study by considering 
	point sets with multiple roots and required monotone
	paths from each root to any other point in the point set. 
	Moreover, he showed how to realize an arbitrary rooted
	tree as a rooted MST with respect to a single direction of
	monotonicity~\cite{mastakas21}.
	
	\subparagraph*{Contribution.} The main results in this paper can be summarized as follows:
	\begin{itemize}
		\item Clearly, the $\MMST(S,\{d\})$ problem can be solved in $O(n \log n)$ time by sorting the points in $S$ with respect to $d$.
		For $|{\cal D}|=2$, we show how to solve $\MMST(S, \D)$ in $O(n^2)$ time; see \cref{le:2-monotone-tree-supporting} in \cref{se:2-monotone}.
		
		\item We show how to recognize whether a given tree is 1- or 2-directional monotone; see \cref{thm:monot_interval_comp,thm:2d_monotone_recognition}, respectively. Moreover, we give a compact representation of all sets \D such that the given tree is \D-monotone.
		
		\item We provide a characterization of $\D$-monotone spanning trees; see \cref{lem:Branches_large_Wuv}. Based on it, we show that the $\MMST(S,\mathcal{D})$ problem belongs to the complexity class XP (``slicewise polynomial'') when parameterized by the number $k=|{\cal D}|$ of directions; namely, we present an algorithm that solves the $\MMST(S,\mathcal{D})$ problem in $O(f(k)n^{2k-1}\log n)$; see \cref{thm:general-k}.
		\item Regarding the $\MMST(S,k)$ problem, we describe $O(n^2\log n)$- and $O(n^6)$-time algorithms for $k=1$ (\cref{th:1-monotone-tree}) and $k=2$ (\cref{th:2-monotone-tree}), respectively. For $k \geq 3$, we prove that the problem belongs to XP, namely we describe an $O(f(k)n^{2k(2k-1)}\log n)$-time algorithm based on our characterization; see \cref{thm:directional-k}.
		\item    We bound from below the maximum vertex degree of the
                  MMST.  We show that, in contrast to the MST, whose
                  vertex degree is at most
                  six~\cite{PapadimitriouVazirani84}, for every even
                  integer~$k \geq 2$, there exists a point set~$S_k$
                  and a set $\mathcal{D}_k$ of $k$ directions such that
                  any minimum-length $\mathcal{D}_k$-monotone spanning
                  tree of $S_k$ has maximum vertex degree~$2k$; see \cref{theorem:deg_star}.
	\end{itemize}
	        
The remainder of this paper is structured as follows.
\cref{se:preliminaries} gives preliminary definitions.
\cref{se:properties-paths-trees} describes basic properties of
monotone paths and trees.
\cref{se:1-monotone,se:2-monotone,se:k-monotone} deal with
minimum monotone spanning trees for $k=1$, $k=2$, and $k \geq 3$
directions, respectively.  We investigate the maximum vertex degree of
an MMST in \cref{se:maxdeg}.  We conclude with open problems  in
\cref{se:conclusions}.

	\section{Basic Definitions}\label{se:preliminaries}
	Let $C$ denote the unit circle centered at the origin of~$\mathbb{R}^2$. Any segment oriented from the center of $C$ to a point of $C$ defines a \emph{direction vector} or simply a \emph{direction}.  Two different directions are \emph{opposite} if the two segments that define them belong to the same line. Given a direction~$d$ and a set~$S$ of points in the plane, we say that $S$ is in \emph{$d$-general position} if no two points in~$S$ lie on a line orthogonal to~$d$. If $S$ is in $d$-general position, let $\ord(S,d)$ be the linear ordering of the orthogonal projections of the points of~$S$ on any line parallel to~$d$ and directed as $d$; note that $\ord(S,d)$ is uniquely defined.
	Given a direction~$d$ and a point set $S=\{p_1,\dots,p_n \}$ in $d$-general position,
	we say that the geometric path $\langle p_1,\dots,p_n \rangle$ is
	\emph{$d$-monotone} if  $\ord(S,d)=\langle p_1,\dots,p_n \rangle$ or
	$\ord(S,d)=\langle p_n,\dots,p_1 \rangle$; in this case, all
	projections of the oriented segments $\overrightarrow{p_i p_{i+1}}$
	(for $i \in \{1,\dots,n-1\})$ on a line parallel to $d$ point
        towards the same
	direction.  A path is called \emph{monotone} if it is $d$-monotone
	with respect to some direction~$d$.

	Let $S$ be a finite set of points and let~$\cal D$ be a finite set of directions such that no two of them are opposite.
	A spanning tree $T$ of $S$ is \emph{${\cal D}$-monotone} if for every pair of vertices $\{u,v\}$ of $T$, there exists a direction $d \in {\cal D}$
	for which the unique geometric path from $u$ to $v$ in $T$ is $d$-monotone (which requires, in particular, that the subset of points on the path from $u$ to $v$ is in $d$-general position).
	A \emph{minimum $\mathcal D$-monotone spanning tree} of~$S$ 
	is a $\mathcal D$-monotone spanning tree of~$S$ of minimum
	length among all $\mathcal D$-monotone spanning trees of~$S$; we denote by $\MMST(S,\mathcal{D})$ the problem of computing such a tree.  
	For a positive integer~$k$, we say that a spanning tree~$T$ of~$S$
	is \emph{$k$-directional monotone} if there exists a set~$\cal D$
	of~$k$ directions such that $T$ is $\cal D$-monotone.  A \emph{minimum
		$k$-directional monotone spanning tree} of~$S$ 
	is a $k$-directional monotone spanning tree of~$S$ of minimum length among all
	$k$-directional monotone spanning trees of~$S$; we denote by $\MMST(S,k)$ the problem of computing such a tree.   
	To solve this problem, it turns out that it is sufficient to consider only sets ${\cal D}$ of directions such that $S$ is in \emph{${\cal D}$-general position}, i.e., $S$ is in $d$-general position for every $d \in \mathcal{D}$.

	Given two points $u$ and $v$, let $l_{u,v}$ be the line passing
	through~$u$ and~$v$.  Given a direction $d$ and a point $x$, let $d(x)$ be the
	line parallel to~$d$ passing through $x$ and let $\overline{d}$
	be the direction orthogonal to~$d$ obtained by rotating $d$ counterclockwise by an angle of~$90^\circ$.
	Accordingly, $\overline{d}(x)$ is the line orthogonal to~$d(x)$ and
	$\overline{l_{u,v}}(x)$ is the line orthogonal to~$l_{u,v}$ passing through~$x$. 
	Given two vertices~$u$ and~$v$ of a geometric tree~$T$, we denote by $T(u,v)$
	the path of~$T$ from~$u$ to~$v$.
	
	\begin{figure}[tb]
          \centering
          \includegraphics[page=1]{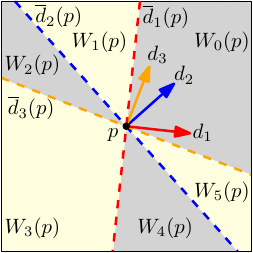}
          \caption{The set $\mathcal{W}_\mathcal{D}(p)$ for the point $p$
            and the set $\D=\{d_1,d_2,d_3\}$ of $k=3$ directions.}
          \label{fi:wedges}
	\end{figure}
	
	Given a sorted\footnote{We always assume that the set \D of directions is sorted with respect to the directions' slopes.} set $\mathcal{D}=\{d_1,d_2, \ldots, d_k\}$ of $k \ge 1$
	(pairwise non-opposite) directions and a point $p$ in the plane,
	let  $\mathcal{W}_\mathcal{D}(p)=\{W_0(p), W_1(p), \ldots, W_{2k-1}(p)\}$
	be the set of $2k$ wedges determined by the lines
	$\overline{d_1}(p), \overline{d_2}(p), \dots, \overline{d_k}(p)$.
	See \cref{fi:wedges} for an illustration, where $k=3$.
	We fix the numbering of the wedges by starting with
	an arbitrary wedge~$W_0(p)$ and then continue with $W_1(p), W_2(p), \dots, W_{2k-1}(p)$
	in counterclockwise order around the origin $p$.
	Whenever we refer to a wedge $W_i(p)$ for some
	integer~$i$, we assume that~$i$ is taken modulo~$2k$. 
	If $p$ coincides with the origin $o$ of~$\mathbb{R}^2$, we just write $\mathcal{W}_\mathcal{D}=\{W_0, W_1, \ldots, W_{2k-1}\}$ instead of $\mathcal{W}_\mathcal{D}(o)=\{W_0(o), W_1(o), \ldots, W_{2k-1}(o)\}$.

	\section{Properties of Monotone Paths and Trees}
	\label{se:properties-paths-trees}
	
	We now review some basic properties of monotone paths and trees.

		\begin{property}
			\label{pr:monotone-half-plane}
			Let~$S$ be a set of points, and let~$T$ be a spanning tree of~$S$.
			Let $x$ be a vertex of $T$, let $u$ and $v$ be two neighbors of $x$
			in $T$, and let $d$ be a direction such that $S$ is in $d$-general position. 
			If $u$ and $v$ lie in the same half-plane determined by $\overline{d}(x)$, then the path between
			$u$ and $v$ in $T$ is not $d$-monotone.
		\end{property}
	
	\begin{proof}
		It is immediate to see that in the linear ordering $\ord(S,d)$, the
		projection of $v$ either precedes or follows both the projections of
		$u$ and $x$. Hence, the path $\langle u, x, v \rangle$ is not
		$d$-monotone (see \cref{fi:monotone-half-plane}).
	\end{proof}
	
	\begin{figure}[hb]
		\centering
		\includegraphics[page=1]{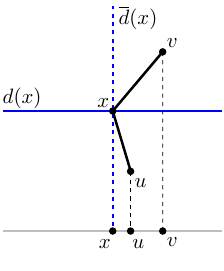}
		\caption{The path $\langle u,x,v \rangle$ is not $d$-monotone}
		\label{fi:monotone-half-plane}
	\end{figure}
	
	The next property generalizes Property~\ref{pr:monotone-half-plane}. It concerns the wedges formed by a set of $k > 1$ directions 
	(in contrast to the half-plane formed by the perpendicular to a single direction) 
	and two arbitrary points in the same wedge
	(in contrast to just two neighbors of $x$).
	
		\begin{property}
			\label{pr:diff_wedges}
			Let~$S$ be a set of points, let $T$ be a spanning tree of~$S$, and let 
			$\cal D$ be a set of $k$ (pairwise non-opposite) directions.
			Let $x$, $u$, and $v$ be points in~$S$ such that $x \in T(u,v)$.
			If $u$ and $v$ lie in the same wedge in $\W_\D(x)$, then the
			path~$T(u,v)$ is not \D-monotone. 
		\end{property}

	\begin{proof}
		Let $S'$ be the set of vertices of $T(u,v)$.  For the path $T(u,v)$ to be monotone
		with respect to some direction $d \in \D$, it must
		hold that $u$ and $v$ are located in different half-planes of
		$\overline{d}(x)$ so that  $x$ appears between $u$ and $v$ in
		$\ord(S',d)$.  This is not possible, however, since~$u$ and~$v$
		lie in the same wedge in $\mathcal{W}_{\D}(x)$.
	\end{proof}

		\begin{property}
			\label{pr:max-degree-2k}
			Let~$S$ be a set of points, and
			let $\cal D$ be a set of $k$ (pairwise non-opposite) directions.
			If~$T$ is a $\cal D$-monotone spanning tree of~$S$, then
			$\Delta(T) \leq 2k$.
		\end{property}

	\begin{proof}
		Let~$x$ be an arbitrary vertex of~$T$.  The set of lines
		$\{\overline{d}(x) \colon d \in \mathcal{D}\}$ partitions the plane
		into $2k$ wedges with apex~$x$.  By
		Property~\ref{pr:diff_wedges}, if two neighbors $u$ and $v$
		of~$x$ lie in the same wedge, then they lie in the same halfplane
		with respect to every direction~$d$ in $\mathcal{D}$.  Hence,
		the path~$\langle u,x,v \rangle$ is not monotone with
		respect to any direction in~$\mathcal{D}$.  Since $T$ is
		$\cal D$-monotone, it follows that no two neighbors of~$x$ lie in
		the same wedge with apex~$x$, which implies that $\deg(x) \le 2k$.
	\end{proof}

		\begin{lemma}
			\label{le:leaves-different-quadrants}
			Let $S$ be a set of points, let~$\mathcal{D}=\{d_1, d_2, \dots, d_k\}$
			be a set of $k$ (pairwise non-opposite) directions such that $S$ is
			in ${\cal D}$-general position, and let~$T$ be a $\cal D$-monotone
			spanning tree of~$S$.  Let $u$ and $v$ be two leaves of a subtree
			$T'$ of $T$. 
			If $u'$ and $v'$ are the vertices adjacent to $u$ and $v$ in $T'$, respectively,  
			then there exist $i, j \in \{0,1,\ldots, 2k-1 \}$ with $i \ne j$ such that $u \in W_i(u')$  and $v \in W_j(v')$.
		\end{lemma}
	
	\begin{proof}
		Assume for the sake of contradiction that $i=j$. If $u'$ and $v'$ coincide, then,  by Property~\ref{pr:monotone-half-plane},  the path between $u$ and $v$ in $T'$ is not $d$-monotone with respect to any direction $d \in \D$; a contradiction. Suppose vice versa that $u'  \neq v'$.
		Consider the path $T(u,v)$ from $u$ to $v$ and $T(u',v')$ from $u'$ to $v'$. Since $T$ is $\D$-monotone, it follows that $T(u,v)$ is monotone with respect to at least one direction, that is, $T(u,v)$ is $d$-monotone for each   direction $d$ in a non-empty set of directions $\D' \subseteq \D$. Then, path $T(u',v')$ is also $d$-monotone for any direction $d \in \D'$ since it is a subpath of $T(u,v)$.
		Consider any direction $d \in \D'$. 
		$T(u',v')$ is contained in the strip delimited by the  two lines $\overline{d}(u')$ and $\overline{d}(v')$.  Also, one wedge among $W_i(u')$ and $W_i(v')$ intersects this strip, while the other does not; suppose, without loss of generality, that $W_i(u')$ intersects the  strip.  Thus, $u$ and $v$ are on the same side of $\overline{d}(u')$. Then, 
		by Property~\ref{pr:diff_wedges}, applied for a single direction,  it follows that  path~$T(u,v)$ is not $d$-monotone; a clear contradiction.
	\end{proof}
	
	The following property is an immediate consequence of \cref{le:leaves-different-quadrants}.  
	
	\begin{property}\label{pr:2k-leaves}
		Let $S$ be a set of points, let ${\cal D}=\{d_1, d_2, \ldots, d_k\}$ be a set of $k$ (pairwise non-opposite) directions such that $S$ is in ${\cal D}$-general position. If $T$ is a $\D$-monotone spanning tree of $S$, then $T$ has at most $2k$ leaves.
	\end{property}
	
	Consider a directed  geometric path $P=\langle p_1,\dots,p_n \rangle$ and let $c_i(P)$ denote the oriented segment starting from the origin of~$\mathbb{R}^2$, and that is parallel to and with the same orientation as segment $\overrightarrow{p_i p_{i+1}}$ of $P$, $1\leq i <n$. 
	Define the \emph{sector of  directions} of path $P$, denoted by
	$\sec(P)$, to be the smallest sector  of the unit circle that includes
	all oriented segments $c_i(P)$, $1\leq i <n$. Refer to
	\cref{fig:path_1,fig:path_2}.  Angelini et
	al.~\cite{AngeliniCBFP12} give the following characterization.
	
	\begin{lemma}[\cite{AngeliniCBFP12}]
		\label{lemma:Angelini}
		Let $P$ be a directed geometric path.  Then $P$ is monotone if and
		only if the angle of its sector of directions $\sec(P)$ is smaller
		than $\pi$.
	\end{lemma}
	
	\begin{figure}[tb]
		\centering
		\begin{subfigure}{.48\linewidth}
			\centering
			\includegraphics[page=1]{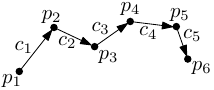}
			\subcaption{}
			\label{fig:path_1}
		\end{subfigure}
		\begin{subfigure}{.48\linewidth}
			\centering
			\includegraphics[page=2]{path}
			\subcaption{}
			\label{fig:path_2}
		\end{subfigure}
		\caption{(a) A directed geometric path and (b) its sector of
			directions $\sec(P)$ (in dark gray).}
	\end{figure}
	
	While \cref{lemma:Angelini} can be used to recognize monotone paths, it does not specify a direction for which the path is monotone. The next lemma fills this gap by indirectly specifying a range of directions for which a given path is monotone.

		\begin{lemma}
			\label{lemma:monotone-path-range}
			Given a direction~$d$, a monotone directed geometric path~$P$ is
			$d$-monotone if and only if $\overline{d}(o)$ does not intersect
			$\sec(P)$, where $o$ is the~origin of~$\mathbb{R}^2$.
		\end{lemma}
	
	\begin{proof}
		Assume first that $P=\langle p_1,\dots,p_n \rangle$ is
		$d$-monotone. Suppose by contradiction that $\overline{d}(o)$
		intersects $\sec(P)$; refer to \cref{fig:range_proof_1}. Let $c_l$
		and $c_r$, $1 \leq l< n$ and $1 \leq r < n$, be the oriented
		segments of the unit circle that delimit the sector of directions
		$\sec(P)$. Then, the projections of the oriented segments
		$\overrightarrow{p_l p_{l+1}}$ and $\overrightarrow{p_r p_{r+1}}$ on
		line $d(o)$ point in opposite directions. This is a clear
		contradiction since all the projections of the oriented segments
		$\overrightarrow{p_i p_{i+1}}$, $1\leq i <n$, of a monotone path
		point in the same direction. Also, note that in the boundary case
		where $\overline{d}(o)$ overlaps with $c_l$ or $c_r$ (or both), path
		$P$ cannot be monotone since the projections of at least two of its
		points on $d(o)$ coincide; another contradiction.
		
		Assume now that $\overline{d}(o)$ does not intersect $\sec(P)$.
		Consider three consecutive path points $p_{i-1}, ~p_i, ~p_{i+1}$,
		$1<i<n$, and let the unit circle be centered at point $p_i$; refer
		to \cref{fig:range_proof_2}. As points $p_{i-1}$ and $p_{i+1}$ are
		on opposite sides of line $\overline{d}(p_i)$, the projections of
		the oriented segments $\overrightarrow{p_{i-1} p_{i}}$ and
		$\overrightarrow{p_{i} p_{i+1}}$ on line $d(p_i)$ point in the same
		direction. Thus, path $P$ is monotone.
	\end{proof}
	
	\begin{figure}[t]
		\centering
		\begin{subfigure}{.48\linewidth}
			\centering
			\includegraphics[page=1]{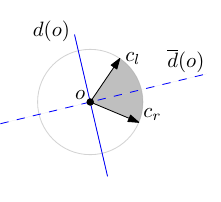}
			\subcaption{}
			\label{fig:range_proof_1}
		\end{subfigure}
		\begin{subfigure}{.48\linewidth}
			\centering
			\includegraphics[page=2]{monotonicity_range}
			\subcaption{}
			\label{fig:range_proof_2}
		\end{subfigure}
		\caption{(a) Line $\overline{d}(o)$ intersects $\sec(P)$.  (b)
			Line $\overline{d}(p_i)$ does not intersect $\sec(P)$.}
	\end{figure}
	
	If a path $P$ is monotone, let $\mathcal{I}(P)$ be the set of
	directions such that $P$ is $d$-monotone.  We call $\mathcal{I}(P)$
	the \emph{monotonicity interval} of~$P$.  By
	\cref{lemma:monotone-path-range}, $\mathcal{I}(P)$ is an open sector
	of the unit circle.  Moreover, given the sector of directions of~$P$,
	$\mathcal{I}(P)$ can be computed in constant time. 
	This yields the following result.
	
	\begin{theorem} 
		\label{thm:monot_interval_comp}
		Given a directed geometric path $P$ on $n$ points, we can check whether $P$ is monotone and, if so, compute $\mathcal{I}(P)$ in $O(n)$ time.
	\end{theorem}

	\section{Minimum 1-Directional Monotone Spanning Trees}
	\label{se:1-monotone}
	
	In this section, we treat 1-directional monotone spanning
	trees. By \cref{pr:max-degree-2k} such a spanning tree is a path and thus \cref{thm:monot_interval_comp}
	immediately implies the following recognition result. 
	
	\begin{corollary}\label{co:1d_monotone_recognition}
		Given a geometric tree $T$ on a set $S$ of $n$ points, we can decide
		in $O(n)$ time whether $T$ is 1-directional monotone and, in the positive case, we
		can specify the set~$\cal D$ of all directions such that $T$
		is $d$-monotone for every $d \in {\cal D}$.
	\end{corollary}
	
	We now prove the main result of this section. 
	
	\begin{theorem}
		\label{th:1-monotone-tree}
		Given a set $S$ of $n$ points, a solution to the $\MMST(S,1)$ problem can be computed in $O(n^2 \log n)$ time.
	\end{theorem}
	
	\begin{proof}
		Based on \cref{pr:max-degree-2k}, any 1-directional monotone spanning tree of $S$ is necessarily a path. For any given direction~$d$ such that $S$ is in $d$-general position, consider $\ord(S,d)$. If we connect every two points of~$S$ whose
		projections are consecutive in $\ord(S,d)$, we uniquely define a
		$d$-monotone spanning path of~$S$. Note that, for two distinct
		directions $d$ and $d'$, the $d$-monotone spanning path might
		coincide with the $d'$-monotone spanning path.  We describe an
		$O(n^2 \log n)$-time algorithm that solves $\MMST(S,1)$; it
		considers all (and only) the distinct 1-directional monotone spanning paths of $S$ and returns one of minimum length.
		
		Assume, for now, that the point set $S$ does not contain three or
		more collinear points and, moreover, no two pairs of points lie on parallel lines. Later on, we will describe how to deal with
		an arbitrary point set. 
		Let $o$ be a point in the plane such that $o  \notin S$, and define set  $\LL$ to consist of   the $h= {n \choose 2}$ lines $\overline{l_{u,v}}(o),~u, v \in S$ with $u \neq
		v$, passing through point $o$ (see the dashed lines in
		\cref{fi:projections_simple}). We can think of  point $o$ as
		being the origin of $\mathbb{R}^2$. Then, these $h$ lines partition
		the unit circle into $2h$ sectors. Start from an arbitrary sector
		and let $d_1$ be the direction that bisects it. Consider then the
		next sector in counterclockwise order and let $d_2$ be the direction
		that bisects it. By continuing in this manner, we can define a
		circular sequence $\sigma = \langle d_1,d_2, \ldots, d_h \rangle$ of
		$h$ pairwise non-opposite directions (see the red direction in
		\cref{fi:projections_simple}).  This construction of the
		direction set $\sigma$ was outlined by Goodman and
		Pollack~\cite{GOODMAN1980220}.
		They showed that, for every $i=1, \dots, h-1$, the linear
		orderings $\ord(S,d_i)$ and $\ord(S,d_{i+1})$ differ exactly for the
		positions of two consecutive points $p$ and $p'$, namely $p$
		immediately precedes $p'$ in $\ord(S,d_i)$, while $p$ immediately
		follows $p'$ in $\ord(S,d_{i+1})$.  By construction, for each
		direction $d \in \sigma$, $S$ is in $d$-general position.  Also,
		$\sigma$ can be computed in $O(n^2 \log n)$ time by sorting the
		distinct slopes of the $n \choose 2$  lines that are defined by
		point pairs in~$S$.
		
		\begin{figure}[tb]
			\centering
			\includegraphics[page=1]{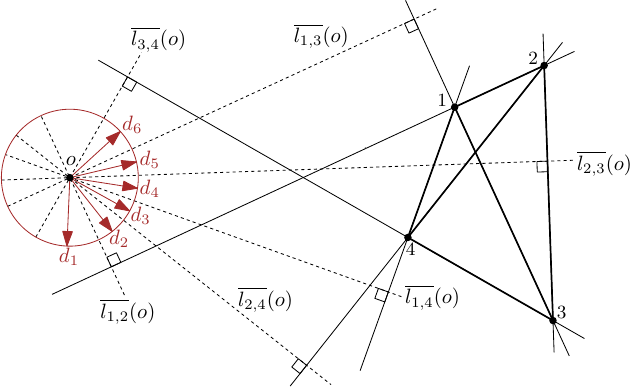}
			\caption{Point set $S= \{1,2,3,4\}$ and illustration of the  computation of a set of distinct directions over which the minimum monotone spanning path of $S$ is computed.}
			\label{fi:projections_simple}
		\end{figure} 
		
		Obviously,  each direction $d$ in $\sigma$ defines a distinct monotone spanning
		path of $S$; to form the path, simply connect the points in order of appearance (of their projections) in $\ord(S, d)$.  Moreover,  for each monotone spanning path of $S$ there exists at least one direction $d$ of $\sigma$ such that this path is $d$-monotone. To see that, 
		simply observe that the boundaries of the  monotonicity interval of the path (as computed in \cref{thm:monot_interval_comp}) are lines in $\LL$ and, thus, at least one direction in $\sigma$ is contained in the monotonicity interval.  
		For $i \in [h]$, let~$P_i$ be the monotone path defined
		by~$\ord(S,d_i)$, and let $\|P_i\|$ be the length of $P_i$. For  implementation purposes, we assume that projections of points of $S$ in $\ord(S,d_i)$ (and, consequently, path $P_i$) are stored in order of appearance in a doubly linked list of \emph{projection objects}.  Moreover, in order to be able to locate the projection of  each point in the sorted list, we maintain with each point of  $S$  a pointer to its projection object.
		Path~$P_1$ and its length~$\|P_1\|$ are easily computed in
		$O(n \log n)$ time by sorting the points in~$S$ with respect to their
		projection on~$d_1$. 
		Then, for each $i \in \{2,\dots,h\}$, path~$P_i$
		and its length~$\|P_i\|$ are computed in $O(1)$ time from~$P_{i-1}$
		and $\|P_{i-1}\|$, by just updating the segments incident to the two
		points $p$ and $p'$ of $S$ that exchange their position passing from
		$P_{i-1}$ to $P_i$.  
		Note that it is not hard to  identify~$p$
		and~$p'$, as they are the points that define the line
		$\overline{l_{p,p'}(o)}$ that separates $d_{i-1}$ and $d_i$ in our
		construction; simply  associate with each line in our construction the two points that define it.
		Hence, the algorithm computes all paths
		defined by $\sigma$ and the lengths of these paths in $O(n^2 \log n)$
		time.

		\begin{figure}[tb]
			\centering
			\includegraphics[page=2]{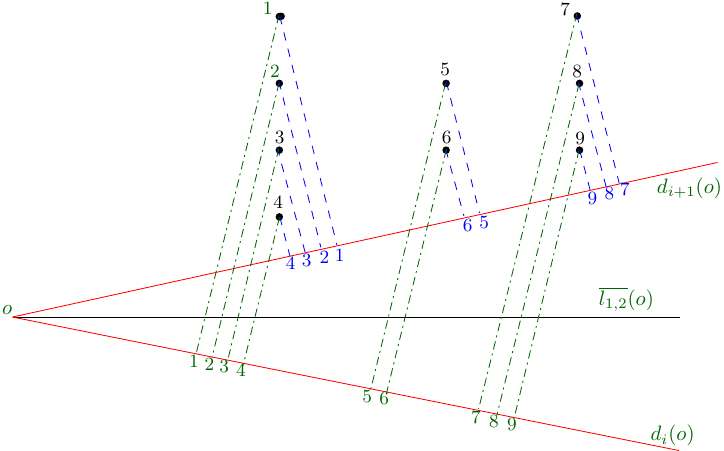}
			\caption{Point set $S=\{1,2, \ldots, 9 \}$ consists of three sets of collinear points that form pairs having the line passing through them perpendicular to $\overline{l_{1,2}}(o)$ and its orderings $\ord(S, d_{i})$  and $\ord(S, d_{i+1})$ on $d_i(o) $ and $d_{i+1}(o)$, respectively.}
			\label{fi:projections_collinear}
		\end{figure}
		
		We now describe how to deal with the case where point set $S$
		contains pairs of points lying on parallel lines. Note that these
		pairs of points may be lying on the same line, resulting in
		more than three collinear points.  For an example point set, refer
		to \cref{fi:projections_collinear}. 
		
		We again  compute  set  $\LL$ consisting of  the $h={n \choose 2}$ lines $\overline{l_{u,v}}(o)$
		for every $u,v \in S$ with $u \neq v$, and sort the pairs of points
		with respect to the slopes of the corresponding lines in $\LL$. 
		In contrast with the ``simple'' point set examined in the
		previous paragraphs, we now end up with a smaller set of $h'$
		distinct slopes, where $h'<h$.  In addition, these new distinct slopes
		partition the
		$h$ pairs of points into $h'$ disjoint sets $E_1, E_2, \ldots, E_{h'}$, each containing pairs of points having identical slopes for their corresponding 
		lines in $\LL$.
		From each set $E_i$, we can select an arbitrary pair of points, say
		$(u_i, v_i)$, $1 \leq i \leq h'$, as the \emph{representative pair}.  We note
		that sets $E_i,~1\leq i \leq h'$, can be computed in $O(n^2 \log n)$
		time by simply sorting the pairs of points with respect to the slopes of their corresponding lines in $\LL$;
		pairs of identical
		slope  end up consecutive after sorting.  We also observe that
		each set $E_i$ with $1 \leq i \leq h'$ induces a graph $G_i = (V_i,E_i)$
		whose vertex set~$V_i$ contains exactly the points involved in the pairs
		of~$E_i$.  In addition, observe that each graph~$G_i$ consists of $k_i$
		connected components each of which is a clique and corresponds to
		points lying on the same line perpendicular to $\overline{l_{u_i,v_i}}(o)$ and, moreover, the projection of the points of each of these 
		$k_i$ connected components on $d_i(o)$ and $d_{i+1}(o)$ do not overlap.
		In the example of \cref{fi:projections_collinear}, we
		have that $V_i= \{ 1, \ldots, 9\}$ and the three connected components
		(cliques) of $V_i$ are $V_{i,1}= \{1,2,3,4\}$,
		$V_{i,2}= \{5,6\}$ and $V_{i,3}= \{7,8,9\}$.  Note that the connected
		components of all graphs $G_i$ with $1\leq i \leq h'$ can be computed
		using depth first search in $O(n^2)$ total time since there are
		exactly ${n \choose 2}$ edges in all graphs together.

		Thus, for the case of a point set  containing pairs of points that lie on parallel lines, we can define the set $\sigma'= \{d_1, d_2, \ldots, d_{h'} \}$ of directions  by considering only the $h'$ distinct slopes of the lines in $\LL$.
		It remains, however, to describe how we  compute for two consecutive  arbitrary directions $d_i$ and $d_{i+1}$, $1\leq i <h'$,  $\ord(S, d_{i+1})$ from $\ord(S, d_{i})$. Let $\overline{l_{u_i, v_i} }(o)$ be the line separating  $d_i$ and $d_{i+1}$ where $(u_i, v_i)$ is the representative pair of $E_i$. 
		
		Observe that all collinear points lying on a line perpendicular to
		$\overline{l_{u_i, v_i} }(o)$ appear in reverse order in $d_{i+1}(o)$
		compared to the order they appear in $d_i(o)$; see
		\cref{fi:projections_collinear}.  Thus, in order to compute
		$\ord(S, d_{i+1})$ from $\ord(S, d_{i})$, we have simply to identify
		these points.  Of course, this has to be repeated for all different
		perpendicular lines to $\overline{l_{u_i, v_i} }(o)$ that contain at
		least a pair of points.  However, we have already computed this
		information.  The sets of collinear points perpendicular to
		$\overline{l_{u_i, v_i} }(o)$ correspond to the vertex sets $V_{i,j}$
		with $1 \le j \leq k_i$ of the $k_i$ connected components of the
		graph~$G_i$.  Thus, the extra cost for computing $\ord(S, d_{i+1})$
		from $\ord(S, d_{i})$ is $O(V_i)$, which amounts to the reversion of
		the order of the points in the projections.  We conclude that all such
		reversions can be computed in $O(n^2)$ total time since the total
		number of edges in all computed graphs equals~${n \choose 2}$.
	\end{proof}

	\section{Minimum 2-Directional Monotone Spanning Trees}\label{se:2-monotone}
	
	In this section, we prove that a solution to the $\MMST(S,2)$ problem for a given point set~$S$ can be computed in polynomial time (\cref{th:2-monotone-tree}); a building block for this result is a procedure that computes a minimum monotone spanning tree of $S$ for a prescribed set ${\cal D}$ of two directions, under the assumption that $S$ is in ${\cal D}$-general position (\cref{le:2-monotone-tree-supporting}).
	In the last part of the section, we prove that recognizing 2-directional monotone spanning trees can be done in linear time (\cref{thm:2d_monotone_recognition}); this result extends the result of \cref{co:1d_monotone_recognition}.  
	
	An \emph{embedding} of a tree specifies,
	for every vertex of the tree, 
	a clockwise circular order
	of the edges incident to the vertex. A tree with a given
	embedding is an \emph{embedded tree}. 
	Given two embedded trees $T_1$ and $T_2$, we say that they \emph{have the same topology} if they have the same embedding after contracting all degree-2 vertices. 
We start by defining, for a spanning tree $T$ of $S$, four different
	\emph{topologies of~$T$ with respect to~$\mathcal{D}$}, where $\D$ is a set of two directions $d_1$ and $d_2$ such that $S$ is in $\D$-general position.  We will show
	that every \D-monotone spanning tree of $S$ has one of these topologies with respect to \D; see \cref{le:2-monotone-structures} for the characterization and
	Fig.~\ref{fi:2-monotone-structure} for a schematic illustration. 
	For $\D=\{d_1, d_2 \}$ and a point $p$ in the plane, recall that $\W_\D(p)= \{ W_0(p), W_1(p), W_2(p), W_3(p)\}$ denotes the set of wedges formed at $p$ by the lines orthogonal to the directions in $\D$.
	
	\begin{enumerate}
		\item \emph{$\mathcal{D}$-path}: $T$ is a path that is $d_1$-monotone, or $d_2$-monotone, or both.
		
		\item \emph{single-degree-4 $\mathcal{D}$-tree}: $T$ consists of a degree-4 vertex~$v$ and four paths emanating from~$v$ such that each path lies in a distinct wedge in $\mathcal{W}_\mathcal{D}(v)$ and is both $d_1$- and $d_2$-monotone.
		
		\item \emph{single-degree-3 $\mathcal{D}$-tree}: $T$ consists of a degree-3 vertex~$v$ and three paths emanating from $v$ such that, for some $i \in \{0,1,2,3\}$, one path lies in the wedge~$W_i(v)$ and one
		in~$W_{i+1}(v)$, these two paths are both $d_1$- and
		$d_2$-monotone, and the third path connects all points in
		$W_{i+2}(v) \cup W_{i+3}(v)$ and is $d$-monotone, where $d \in\cal D$ is the direction orthogonal to the line that separates $W_i(v) \cup W_{i+1}(v)$ from $W_{i+2}(v) \cup W_{i+3}(v)$.
		
		\item \emph{double-degree-3 $\mathcal{D}$-tree}: $T$ consists of two degree-3 vertices~$u$ and~$v$ and five paths such that, for some $i \in \{0,1,2,3\}$, one path lies in~$W_i(u)$, one in~$W_{i+1}(u)$, one in~$W_{i+2}(v)$, and one in~$W_{i+3}(v)$; these four paths are both $d_1$- and $d_2$-monotone, and the fifth path connects all points in the infinite strip $\mathbb{R}^2 \setminus (W_i(u) \cup W_{i+1}(u) \cup W_{i+2}(v) \cup W_{i+3}(v))$ and is $d$-monotone, where $d \in \cal D$ is the direction orthogonal to the two lines delimiting the strip.
	\end{enumerate}
	
	\begin{figure}[tb]
		\begin{subfigure}{.31\textwidth}
			\centering
			\includegraphics[page=1]{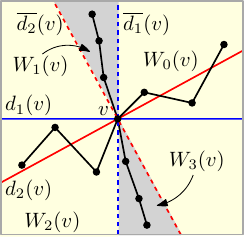}
			\subcaption{single-degree-4}
			\label{fi:single-degree-4}
		\end{subfigure}
		\hfill
		\begin{subfigure}{.31\textwidth}
			\centering
			\includegraphics[page=2]{types}
			\subcaption{single-degree-3}
			\label{fi:single-degree-3}
		\end{subfigure}
		\hfill
		\begin{subfigure}{.31\textwidth}
			\centering
			\includegraphics[page=3]{types}
			\subcaption{double-degree-3}
			\label{fi:double-degree-3}
		\end{subfigure}
		\caption{Different types of $\mathcal{D}$-trees, i.e., topologies of a spanning tree $T$ that is monotone with respect to a set $\mathcal{D}=\{d_1,d_2\}$ of two
			directions.}
		\label{fi:2-monotone-structure}
	\end{figure}
	
	The auxiliary properties stated by \cref{le:four-leaves-monotone,le:degree-3-monotone} are used to prove the characterization of \cref{le:2-monotone-structures}.

		\begin{lemma}
			\label{le:four-leaves-monotone}
			Let $S$ be a set of points, let ${\cal D}=\{d_1, d_2\}$ be a set of two (non-opposite) directions such that $S$ is in ${\cal D}$-general position, and    let $T$ be a $\cal D$-monotone spanning tree of $S$.  If $T$ has exactly four leaves, then every path from a leaf to its nearest vertex of degree at least three is both $d_1$- and $d_2$-monotone.
		\end{lemma}

	\begin{proof}
		Suppose that at least one of the four paths connecting a leaf to a vertex of degree larger than two is monotone in only one direction, say $d_1$. Call this path $P$. Then there are three vertices $u$, $v$, and $w$ that are consecutive in this order along $P$ (starting from the leaf of $P$) and such that $u$ and $w$ are on the same side of $\overline{d}_2(v)$. This means that $u \in W_{i}(v) \cup W_{i+1}(v)$ for some $i \in \{0,1,2,3\}$, and $v \in W_{i+2}(w) \cup W_{i+3}(w)$. Let $T'$ be the subtree of $T$ obtained by removing all vertices of $P$ from the leaf up to $v$ (excluded). Tree $T'$ has four leaves one of which is $v$. Suppose, w.l.o.g., that $v \in W_{i+2}(w)$; by \cref{le:leaves-different-quadrants} the other three leaves are in wedges $W_i(x)$, $W_{i+1}(y)$, and $W_{i+3}(z)$, for three vertices $x$, $y$ and $z$ (possibly coincident). But this means that the subtree $T''$ obtained by removing all vertices of $P$ from the leaf up to $u$ (excluded) has two leaves either in wedges $W_{i}(v)$ and $W_i(x)$ or in wedges $W_{i+1}(v)$ and $W_{i+1}(y)$, thus contradicting \cref{le:leaves-different-quadrants}.
	\end{proof}

		\begin{lemma}
			\label{le:degree-3-monotone}
			Let $S$ be a set of points, let ${\cal D}=\{d_1, d_2\}$ be a set of two (non-opposite) directions such that $S$ is in ${\cal D}$-general position, and let $T$ be a $\cal D$-monotone spanning tree of $S$.  If $T$ has a degree-3 vertex $u$, then two of the paths emanating from $u$ and leading to a leaf are both $d_1$- and $d_2$-monotone.
		\end{lemma}
	
	\begin{proof}
		By \cref{pr:2k-leaves}, $T$ has at most four leaves and hence at least two of the paths emanating from $u$ end with a leaf. The third path either also ends with a leaf or it connects $u$ to another vertex $v$ of degree three. If $T$ has two vertices $u$ and $v$ of degree three, then $T$ has exactly four leaves and the statement follows for both $u$ and $v$ by \cref{le:four-leaves-monotone}. So, assume that three paths $P_1$, $P_2$ and $P_3$, each ending with a leaf, emanate from $u$. 
		Assume for a contradiction that two of these paths, say $P_1$ and
		$P_2$, are monotone with respect to a single direction. 
		Clearly, the direction
		must be the same for both paths, otherwise $P_1 \cup P_2$ is not
		monotone.  Let $d_1$ be such a direction.  Since $P_1 \cup P_2$ is
		monotone only in the $d_1$-direction, the line $\overline{d_1}(u)$
		separates $P_1$ from $P_2$.  The path $P_3$ is either in the same
		side of $\overline{d_1}(u)$ as $P_1$ or in the same side as~$P_2$.
		Assume that it is in the same side as $P_1$ (the other case
		is analogous).  This means that $P_1 \cup P_3$ is not
		$d_1$-monotone.  Nonetheless, $P_1 \cup P_3$ cannot be
		$d_2$-monotone, as $P_1$ is not $d_2$-monotone; a contradiction.
	\end{proof}
	
	We are now ready to give the following characterization.

		\begin{lemma}
			\label{le:2-monotone-structures}
			Let $S$ be a set of points and let ${\cal D}=\{d_1, d_2\}$ be a set of two (non-opposite) directions such that $S$ is in  ${\cal D}$-general position.  
			Then, a tree $T$ is a $\cal D$-monotone spanning tree of $S$ if and only if $T$ is either a $\mathcal{D}$-path, a single-degree-4 $\mathcal{D}$-tree, a
			single-degree-3 $\mathcal{D}$-tree, or a double-degree-3 $\mathcal{D}$-tree.
		\end{lemma}

	\begin{proof}
		It is easy to see that if $T$ has one of the four topologies defined above, then $T$ is $\cal D$-monotone. 
		Suppose vice versa that $T$ is $\cal D$-monotone. By \cref{pr:max-degree-2k},
		each vertex of $T$ has degree at most four and, by \cref{pr:2k-leaves}, $T$ has at
		most four leaves. If $T$ has exactly four leaves then, by \cref{le:four-leaves-monotone}, the four paths connecting these leaves to vertices of degree larger than two are both $d_1$- and $d_2$-monotone. By \cref{pr:diff_wedges},
		each of these paths is in a wedge with a different index. Thus $T$ is either a single-degree-4 $\mathcal{D}$-tree, or a double-degree-3 $\mathcal{D}$-tree
		with respect to $\mathcal{D}$. If $T$ has only three leaves, then it has a degree-3 vertex $u$ and by \cref{le:degree-3-monotone}, two paths $P_1$ and $P_2$ connecting $u$ to a leaf are monotone in both directions and thus are contained each in a single wedge with apex $u$.
		Let $P_3$ be the third path.
		We now consider the following two cases.
		\begin{enumerate}
			\item If the two wedges containing $P_1$ and $P_2$ have consecutive indices, then~$P_3$ spans all points in the remaining wedges and is $d$-monotone, where $\overline{d}$ is the line separating the remaining wedges from those containing $P_1$ and $P_2$.
			\item If the two wedges containing $P_1$ and $P_2$ have non-consecutive indices, then $P_3$ is completely contained in one of the other two wedges, in which case it must be monotone in both directions, because the path $P_1 \cup P_3$ is on the same side as $\overline{d}_1(u)$ (or as $\overline{d}_2(u)$) and hence it must be $d_2$-monotone (or $d_1$-monotone), while the path~$P_2 \cup P_3$ is on the same side as $\overline{d}_2(u)$ (or as $\overline{d}_1(u)$) and hence it must be $d_1$-monotone (or $d_2$-monotone).
		\end{enumerate}
		Thus, in both cases,~$T$ is a single-degree-3 $\mathcal{D}$-tree. Finally, if $T$ has only two leaves, then it is a path and is either $d_1$-, or $d_2$-monotone, or both, i.e., it is a~$\mathcal{D}$-path.
	\end{proof}
	
	The following two lemmas are used to prove
	\cref{le:2-monotone-tree-supporting}, which in turn yields
	\cref{th:2-monotone-tree}.  The algorithm that proves
	\cref{le:monotone-path-wedge} is
	reminiscent of the sweep-line algorithm for computing the
	maxima of a set of points~\cite{klp-fmsv-JACM75}.

		\begin{lemma}
			\label{le:monotone-path-wedge} 
			Given a set~${\cal D}=\{d_1,d_2\}$ of two (non-opposite) directions and a point
			set~$S$ in \D-general position, we can compute a table
			$\mathcal{Q}(S,\D)$ such that:
			$(i)$ $\mathcal{Q}(S,\D)$ reports in $O(1)$ time, for
			a query point~$p$ in~$S$ and $i \in \{0,1,2,3\}$, whether
			the points in $W_i(p) \cap S$ form a path that is both
			$d_1$- and $d_2$-monotone. If yes, the length of the
			path is also reported. 
			$(ii)$~$\mathcal{Q}(S,\D)$ has size $O(n)$ and can be computed
			in $O(n \log n)$ time, where $n=|S|$.
		\end{lemma}

	\begin{proof}
		Let $\D=\{d_1,d_2\}$.
		The table $\mathcal{Q}(S,\D)$ simply stores, for each pair
		$(p,i)$, with $p \in S$ and $i \in \{0,1,2,3\}$, the following data:
		(1)~a Boolean flag that is true if and only if the points in
		$W_i(p) \cap S$ form a path that is both $d_1$- and $d_2$-monotone
		and, if the flag is true, (2)~the length $\ell(p,i)$ of the
		corresponding path.  Observe that $\mathcal{Q}(S,\D)$ has size $O(n)$.
		
		To construct $\mathcal{Q}(S,\D)$ in $O(n \log n)$ time, we proceed as
		follows.  
		We first transform the point set~$S$ by an
		affine transformation that maps the $d_1$-coordinates to
		$x$-coordinates and the $d_2$-coordinates to $y$-coordinates.
		Additionally, for each $i \in \{0,1,2,3\}$, we make sure that the wedge~$W_i$ corresponds to the
		first quadrant. This can always be achieved by appropriately
		multiplying all coordinates of some type by~$+1$ or by~$-1$.
		Hence, after our transformation, for any point~$q$ in $S$, the points that were in
		$W_i(q)$ before the transformation are now in $W_1(q)$.

		Sort the points by $x$-coordinate.  This takes $O(n \log n)$ time.
		The rest of the algorithm is iterative; it takes only $O(n)$ time.
		For a point~$p$ in~$S$, let $p_y$ be its $y$-coordinate and let
		$p_x$ be its $x$-coordinate.  To simplify the description of the
		algorithm, we assume that no two points have the same $x$- or
		$y$-coordinate.  We say that a point~$p$ in~$S$ \emph{dominates} a
		point~$r$ if $p_x>r_x$ and $p_y>r_y$.  We say that $p$
		\emph{directly dominates}~$r$ if there is no point~$q$ in~$S$ such
		that $p$ dominates~$q$ and $q$ dominates~$r$.  In other words, given
		a point set~$S'$, there is an $x$- and $y$-monotone path through the
		points in $S'$ if and only if no point in~$S'$ is directly dominated
		by two other points.
		
		Scan the points in~$S$ in order of decreasing $x$-coordinates.  For
		each point~$q$, we do the simple test described below.  If~$q$
		passes the test, we set its flag to true, establish a pointer to the
		next point~$p$ on its $x$- and $y$-monotone path, and set
		$\ell(q,i)=\ell(p,i)+d(q,p)$, where $d(p,q)$ is the Euclidean
		distance of the (untransformed) points~$p$ and~$q$.  If a point
		in~$S$ has no edge directed into it, then we call it \emph{minimal}. 
		At any time, we maintain the minimal point $m$ in~$S$ that currently
		has the largest $y$-coordinate.  We also maintain the point~$m'$
		that has the largest $y$-coordinate among the points in~$S
		\cap W_3(m)$ (that is, among the points to the right and below~$m$).
		Note that $m'$ may or may not be minimal.  In the first iteration,
		we set~$m$ to the rightmost point and set its flag to true.  For
		simplicity, we initially set $m'$ to a dummy point at
		$(\infty,-\infty)$.
		
		For any further iteration, let~$q$ be the current point in~$S$.
		There are three cases depending on the vertical position of~$q$ with
		respect to~$m$ and~$m'$; see \cref{fig:iterative}:
		\begin{enumerate}
			\item \label{enum:case1}
			If $q_y < m'_y$, then $q$ fails the test because the points
			$m$ and $m'$ both dominate it directly; see
			\cref{fig:iterative-case1}.  Set the flag of~$q$ to false.
			
			\item \label{enum:case2}
			If $m'_y < q_y < m_y$, then set the flag of~$q$ to true, establish
			a pointer from~$q$ to~$m$, and set $m = q$; see
			\cref{fig:iterative-case2}.
			
			\item \label{enum:case3}    
			If $m_y < q_y$, then we follow pointers from $m$ to its successors
			as long as the current point is below~$q$; see
			\cref{fig:iterative-case3}.  If the last such point~$p$ has a
			pointer to a point~$r$ in $W_1(q)$, establish a pointer from~$q$
			to~$r$.  Independently of that, set the flag of~$q$ to true, set
			$m = q$, and set $m' = p$.
		\end{enumerate}
		\begin{figure}[tb]
			\begin{subfigure}{.3\textwidth}
                                \centering
				\includegraphics[page=1]{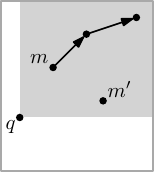}
				\caption{$q_y < m'_y \;(< m_y)$}
				\label{fig:iterative-case1}
			\end{subfigure}
			\hfill
			\begin{subfigure}{.3\textwidth}
                                \centering
				\includegraphics[page=2]{iterative-algorithm}
				\caption{$m'_y < q_y < m_y$}
				\label{fig:iterative-case2}
			\end{subfigure}
			\hfill
			\begin{subfigure}{.3\textwidth}
                                \centering
				\includegraphics[page=3]{iterative-algorithm}
				\caption{$(m'_y <)\; m_y < q_y$}
				\label{fig:iterative-case3}
			\end{subfigure}
			\caption{The three cases that occur in the iterative algorithm.}
			\label{fig:iterative}
		\end{figure}
		
		The algorithm maintains the following invariant throughout the
		algorithm: The point~$m$ is the starting point of a (possibly empty)
		$x$- and $y$-monotone path through all points in $W_1(m) \cap S$.
		Accordingly, the flag of $m$ is always true.
		
		Note that in case~\ref{enum:case1}, neither $m$ nor~$m'$ changes.
		In case~\ref{enum:case2}, $m'$ does not change, whereas~$m$ goes
		down (but stays above $m'$).  Only in case~\ref{enum:case3} the
		point~$m'$ changes.  In that case, $m$ and~$m'$ go up (that is,
		their $y$-coordinates increase).
		
		For the runtime analysis, note that every point has at most one
		pointer to any other point, and we traverse each pointer at most
		once.  This is due to the fact that (a)~the point~$m'$ never goes
		down,~(b)~$m$ is always above~$m'$, and (c)~the pointers that we
		traverse in case~\ref{enum:case3} on the path from~$m$ to~$p$
		(or~$r$) originate in points that will be below~$m'$ after we
		update~$m'$ to~$p$.
		
		For the correctness, we consider the three possible types of wrong
		outcomes of the algorithm and show that each of them leads to a
		contradiction.
		
		First assume that there is a point~$q$ in~$S$ such that the points
		in $W_1(q) \cap S$ form an $x$- and $y$-monotone path, but the
		algorithm set the flag of~$q$ to false.  Suppose that $q$ is the
		first (that is, rightmost) point where the algorithm makes this
		mistake.  But then the flag of the successor~$q'$ of~$q$ on the path
		is true, and the points in $W_1(q) \cap S$ form an $x$- and
		$y$-monotone path starting in~$q'$.  When the algorithm reaches~$q$,
		either $q'$ is a minimal point, so $m=q'$ (case~\ref{enum:case2};
		note that $m$ cannot be above~$q'$ because either $q$ or $q'$ would
		be directly dominated by two points), or $m$ is below~$q$
		(case~\ref{enum:case3}).  However, in both cases the algorithm would
		have added a pointer from~$q$ to~$q'$ and would have set the flag
		of~$q$ to true, contradicting our assumption.
		
		Now assume that there is a point~$q$ in~$S$ that is directly
		dominated by two other points in~$S$, but the algorithm sets the flag
		of~$q$ to true.  Suppose again that~$q$ is the first point where the
		algorithm makes this mistake, and let $p$ and $r$ with $p_y>r_y$ be
		the two points that directly dominate~$q$.  If~$p$ is minimal, then
		either $m=p$ and $m'=r$ and we are exactly in the situation of
		case~\ref{enum:case1}, or $m$ is above~$p$ and/or $m'$ is above~$r$,
		and we are still in case~\ref{enum:case1}.  So if $p$ is minimal,
		the algorithm sets the flag of~$q$ to false, contradicting our
		assumption.  If~$p$ is not minimal, then there is a point~$q'$ to the
		right of (and below) $q$ that has a pointer to~$p$.  But $q'$ would
		be directly dominated by $p$ and~$r$, contradicting our choice
		of~$q$.
		
		Finally, assume that there is a point~$q$ in~$S$ such that
		$W_1(q) \cap S$ contains a point~$q'$ directly dominated by two
		other points in~$S$, but the algorithm sets the flag of~$q$ to true.
		Let~$q'$ be the last such point in $W_1(q) \cap S$ treated by the
		algorithm.  As we have argued above, the algorithm has correctly
		recognized~$q'$ (due to points~$m$ and~$m'$ in $W_1(q')$).  Until the
		algorithm treats~$q$, the points~$m$ and~$m'$ may change, but
		since~$m'$ never goes down and~$m$ stays above~$m'$, both $m$ and
		$m'$ are contained in~$W_1(q)$ (which contains $W_1(q')$).  Hence,
		the algorithm would actually have set the flag of~$q$ to false when
		treating~$q$, contradicting our assumption.
	\end{proof}

		\begin{lemma}
			\label{le:monotone-path-strip}
			Given a direction~$d$ and a point set $S$ in $d$-general
			position, we can compute a table $\mathcal{Q}'(S,d)$ such
			that:
			$(i)$ $\mathcal{Q}'(S,d)$ reports in $O(1)$ time, for a query
			pair of points $\{p,q\}$ in~$S$, the length of the unique
			$d$-monotone path from $p$ to $q$ passing through all
			points in the infinite strip bounded
			by~$\overline{d}(p)$~and~$\overline{d}(q)$.
			$(ii)$ $\mathcal{Q}'(S,d)$ has size $O(n)$ and can be computed
			in $O(n \log n)$ time, where $n=|S|$.
		\end{lemma}
	
	\begin{proof}
		Let $\ord(S,d)=p_1,p_2,\dots,p_n$.  The table $\mathcal{Q}'(S,d)$
		associates with each $i=1,2,\dots,n$ the length $l_i$ of the path
		$p_1,p_2,\dots,p_i$.  Given a pair of points $p=p_i$ and $q=p_j$,
		with $1 \leq i,j \leq n$, the length of the path from $p$ to $q$
		passing through all points in the infinite strip bounded by
		$\overline{d}(p)$ and $\overline{d}(q)$ is $|l_j-l_i|$, which is
		computed in constant time using the values stored in the table at
		indices $i$ and $j$.
	\end{proof}
	
	\begin{lemma}
			\label{le:2-monotone-tree-supporting}
			Let $S$ be a set of $n$ points and let ${\cal D}=\{d_1,d_2\}$ be a
			set of two (non-opposite) distinct directions such that $S$ is in
			${\cal D}$-general position.  There exists an $O(n^2)$-time
			algorithm that computes a minimum $\cal D$-monotone spanning
			tree~of~$S$.
	\end{lemma}

	\begin{proof}
		We give an algorithm that, for each possible topology of
		\cref{le:2-monotone-structures}, checks whether a spanning tree with
		that topology exists, and if so, it computes one of minimum
		length. Among the at most four resulting trees, the algorithm
		returns one of minimum length.
		
		We first set up the data structure $\mathcal{Q}(S,\D)$ mentioned in
		\cref{le:monotone-path-wedge}. Then, we set up the data structures
		$\mathcal{Q}'(S,d_1)$ and $\mathcal{Q}'(S,d_2)$ of
		\cref{le:monotone-path-strip}.  This preprocessing takes
		$O(n \log n)$ total time.
		$\mathcal{Q}'(S,d_1)$ and $\mathcal{Q}'(S,d_2)$ immediately give us
		the lengths of the unique $d_1$- and $d_2$-monotone spanning
		paths. The shorter of the two is a $\D$-path and is stored as a
		candidate for the minimum $\cal D$-monotone spanning tree of~$S$.
		
		Then we go through each point~$p$ in~$S$ and check whether~$p$ can
		be the unique degree-4 node of a single-degree-4 $\mathcal{D}$-tree.
		To this end, we query $\mathcal{Q}(S,\D)$ with~$p$ and with each of
		the four wedges in $\mathcal{W}_\mathcal{D}(p)$.  If the data
		structure returns ``yes'' four times, that is, if the points in each
		of the four wedges form a $d_1$- and $d_2$-monotone path, we add up
		their lengths and compare their sum to the length of the shortest
		single-degree-4 $\mathcal{D}$-tree found so far (if any).
		
		As in the previous case, for the single-degree-3 $\mathcal{D}$-tree
		we go through each point~$p$ in~$S$ and check whether~$p$ can be the
		unique degree-3 node.  We query $\mathcal{Q}(S,\D)$ with~$p$ and
		with each of the four wedges in $\mathcal{W}_\mathcal{D}(p)$.  If
		the data structure returns ``yes'' for a pair of neighboring
		wedges~$W_i(p)$ and~$W_{i+1}(p)$, let~$l_1$ and~$l_2$ be the lengths
		of the paths in~$W_i(p)$ and~$W_{i+1}(p)$, respectively, that are
		both $d_1$- and $d_2$-monotone.  Let $d \in \mathcal{D}$ be the
		direction orthogonal to the line separating~$W_i(p)$
		and~$W_{i+1}(p)$.  We query $\mathcal{Q}'(S,d)$ for the length~$l_3$
		of the $d$-monotone path that starts in~$p$ and goes through all
		points in $W_{i+2}(p) \cup W_{i+3}(p)$.  Then we compare the sum
		$l_1+l_2+l_3$ to the length of the shortest single-degree-3
		$\mathcal{D}$-tree found so~far~(if any).
		
		Finally, we compute a minimum-length double-degree-3
		$\mathcal{D}$-tree, if such a tree exists.  We go through every pair
		$\{p,q\}$ of points in~$S$ and check whether $S$ admits a
		double-degree-3 $\mathcal{D}$-tree whose only two degree-3 vertices
		are~$p$ and~$q$.  To this end, we query the table
		$\mathcal{Q}(S,\D)$ with~$p$ and with each of the four wedges in
		$\mathcal{W}_\mathcal{D}(p)$.  If the table has stored ``true'' for
		a pair of neighboring wedges~$W_i(p)$ and~$W_{i+1}(p)$, then we
		define~$l_1$, $l_2$, and~$d \in \mathcal{D}$ as in the case of the
		single-degree-3 $\mathcal{D}$-tree.  Now we query
		$\mathcal{Q}(S,\D)$ with~$q$ and with the two wedges~$W_{i+2}(q)$
		and~$W_{i+3}(q)$ in $\mathcal{W}_\mathcal{D}(q)$. If the table has
		stored ``true'' for~$W_{i+2}(q)$ and~$W_{i+3}(q)$, then let~$l_3$
		and~$l_4$ be the lengths of the corresponding paths in~$W_{i+2}(q)$
		and~$W_{i+3}(q)$. We query $\mathcal{Q}'(S,d)$ for the length~$l_5$
		of the $d$-monotone path that starts in~$p$, goes through all points
		in the strip delimited by $\overline{d}(p)$ and $\overline{d}(q)$,
		and ends in~$q$.  Then we compare the sum $l_1+\dots+l_5$ to the
		length of the shortest double-degree-3 $\mathcal{D}$-tree that we
		have found so far (if any).
		
		Clearly, after the $O(n \log n)$-time preprocessing, the running
		time of the algorithm is dominated by the time needed to compute the
		shortest double-degree-3 $\mathcal{D}$-tree (if any).  This
		computation requires to iterate over all pairs of points in~$S$,
		but, using $\mathcal{Q}(S,\D)$, $\mathcal{Q}'(S,d_1)$, and
		$\mathcal{Q}'(S,d_2)$, we have only constant work for each pair, and
		hence $O(n^2)$~time in total.
	\end{proof}

			\begin{theorem}
			\label{th:2-monotone-tree}
			Given a set $S$ of $n$ points, a solution to the $\MMST(S,2)$ problem
			 can be computed in $O(n^6)$ time.
				\end{theorem}

	\begin{proof}
		Let $\sigma = \langle d_1, d_2, \dots, d_h \rangle$ be the circular
		sequence of directions with $h \leq {n \choose 2}$ as defined in the
		proof of \cref{th:1-monotone-tree}.  Recall that we
		can compute~$\sigma$ in $O(n^2 \log n)$ time.  By applying
		\cref{le:2-monotone-tree-supporting} to every pair of distinct
		directions $d,d' \in \sigma$, we compute the minimum 
		$\D$-monotone tree for each set $\D=\{d,d'\}$ for which $S$ is in
		$\D$-general position.  Among all these trees, we return one of minimum total length.
		Since there are ${h \choose 2} \in O(n^4)$
		pairs, this takes $O(n^6)$ time.
		
		We now prove that restricting to sets $\D$ for which $S$ is in
		$\D$-general position is sufficient; namely, we prove that if $T$ is
		a $\D$-monotone spanning tree for a set $\D=\{d,d'\}$ and $S$ is not
		in $d$-general position, then there is a set $\D'=\{d',d''\}$ such
		that $S$ is in $d''$-general position and $T$ is
		$\D'$-monotone. This implies that if $T$ is a $\D$-monotone spanning
		tree and $S$ is not in $\D$-general position, then there is also a
		set $\D'$ such that $S$ is in $\D'$-general position and $T$ is
		$\D'$-monotone (with the previous reasoning for one or both the
		directions~in~$\D$).
		Since above, we computed a minimum-length \D-monotone spanning tree $T^\star$ over all sets~\D\ such that $S$ is in \D-general position, the total length of $T$ cannot be less than that of $T^\star$.
		
		Let $T$ be a $\D$-monotone spanning tree, where $\D=\{d,d'\}$. If
		$S$ is not in $d$-general position, the orthogonal projection of $S$
		on a line parallel to $d$ defines a sequence
		$\alpha=\langle p_1,p_2,\dots,p_{n'} \rangle$ of points with
		$n' < n$. Some points of $\alpha$ correspond to the projection of
		multiple points of~$S$; each of these points is called a
		\emph{multiple} point. By slightly rotating~$d$, we can obtain a
		direction $d''$ such that: $(i)$~the projections of the points of~$S$
		that correspond to the same multiple point in~$\alpha$ form a
		consecutive subsequence of $\ord(S,d'')$; and $(ii)$~replacing each
		such subsequence with a single point, we obtain
		$\alpha$.  Finally, we show that $T$ is $\D'$-monotone,
		where $\D'=\{d',d''\}$.  Let~$P$ be a path
		between two points $u$ and $v$ in $T$; if $P$ is $d$-monotone then
		the points of $P$ are in $d$-general position and thus the
		orthogonal projections of all points of~$P$ correspond to distinct
		(non-multiple) points of~$\alpha$.  Hence the points in~$P$ are also
		in $d''$-general position, and $P$ is also $d''$-monotone.
	\end{proof}

		We close this section by turning our attention to the problem of
		recognizing whether a given geometric tree $T$ is 2-directional
		monotone.
		To this end, we extend our solution regarding the recognition problem for one
		direction (see \cref{co:1d_monotone_recognition}) to two directions.

		\begin{theorem}
			\label{thm:2d_monotone_recognition}
			Given a geometric tree $T$ on a set $S$ of $n$ points, we can decide
			in $O(n)$ time whether $T$ is 2-directional monotone.
			In the positive case and if $T$ is not a path, we specify a
			pair of intervals $\mathcal{I}_{1}$ and $\mathcal{I}_{2}$ such that $T$ is
			$\mathcal{D}$-monotone for every~$\mathcal{D}=\{d_1, d_2\}$ with $d_1 \in \mathcal{I}_{1}$ and $d_2 \in \mathcal{I}_{2}$.
		\end{theorem}

	\begin{proof}
		For a tree $T$ to be monotone with respect to two
		(non-opposite) directions, it must adhere to one of the four
		different topologies described in \cref{le:2-monotone-structures}.
		Therefore, if $\Delta(T)>4$, then $T$ is not 2-directional monotone.
		Otherwise, for each topology, we compute the monotonicity interval of each path that connects a pair of leaves; if these intervals are not empty, we check whether they contain a set~\D\ of two directions that fulfills the requirements of the topology listed at the beginning of \cref{se:2-monotone}.  
		
		For two paths $P_1$ and $P_2$ sharing a common
		endpoint, let $P_1P_2$ denote the path resulting
		from their concatenation. Note that the above notation can be
		recursively applied to cover the concatenation of more than two
		paths into a single one. 
		We now distinguish the three cases based
		on~$\Delta(T)$. 
		
		\begin{enumerate}
			\item $\Delta(T)=2$: Then, $T$ is a path.  By
			\cref{thm:monot_interval_comp}, we decide whether $T$ is monotone
			and, if yes, we compute its monotonicity interval
			$\mathcal{I}(T)$.  Then $T$ is $d_1$-monotone for any direction
			$d_1 \in \mathcal{I}(T)$.
			
			\item $\Delta(T)=3$: By \cref{le:2-monotone-structures}, for $T$ to
			be $\mathcal{D}$-monotone, it should be either a single-degree-3
			$\mathcal{D}$-tree or a double-degree-3 $\mathcal{D}$-tree.  We
			distinguish two cases based on the number of degree-3 vertices.
			
			\begin{enumerate}
                                \item {\it $T$ has a single degree-3 vertex~$u$      
                                  (\cref{fig:rec_case2a}):} Then, $T$
                                consists of three paths $P_1$, $P_2$,
                                and $P_3$ that have~$u$ as their
                                common endpoint and that lead to
                                distinct leaves.
				
				For $T$ to be $\mathcal{D}$-monotone, the paths
				$P_1 P_2$, $P_1 P_3$,
				and $P_2 P_3$ have to be monotone.  Their
				monotonicity can be verified by applying
				\cref{thm:monot_interval_comp}.  If they are monotone, let
				$\mathcal{I}_{1,2}$, $\mathcal{I}_{1,3}$, and
				$\mathcal{I}_{2,3}$ be the corresponding monotonicity intervals.
				Observe that if $\mathcal{I}_{1,2}$, $\mathcal{I}_{1,3}$, and
				$\mathcal{I}_{2,3}$ are mutually disjoint, then $T$ is not
				2-directional monotone.  Therefore, at least two of the
				intervals must overlap, for example,
				$\mathcal{I}_1 := \mathcal{I}_{1,2} \cap \mathcal{I}_{1,3} \neq
				\emptyset$.  Now let~$d_1$ be an arbitrary direction in
				$\mathcal{I}_1$, and let~$d_2$ be an arbitrary direction in the
				remaining interval $\mathcal{I}_{2,3}$.  Then, $T$ is
				$\{d_1,d_2\}$-monotone.
				
				\item {\it $T$ has two degree-3 vertices (\cref{fig:rec_case2b}):}
				Then, $T$ consists of two degree-3 vertices~$u$ and~$v$ and five
				paths $P_0, P_1, \dots, P_4$.
				Path $P_0$ has $u$ and $v$ as its endpoints, paths
				$P_2$ and $P_3$ have $u$ as their common
				endpoint, and paths $P_1$ and $P_4$ have
				$v$ as their common endpoint.  Each of the four paths leads to a
				distinct leaf.
				
				Similarly as before, we have to check the six paths 
				$P_1 P_0 P_2$, 
				$P_1 P_0 P_3$,
				$P_4 P_0 P_2$,
				$P_4 P_0 P_3$,
				$P_1 P_4$, and 
				$P_2 P_3$.
				If any of them is not monotone, then $T$ is not 2-directional monotone. 
				The monotonicity of the six paths can be verified by applying
				\cref{thm:monot_interval_comp}.
                                If they are monotone, then let
				$\mathcal{I}_{1,0,2}$, 
				$\mathcal{I}_{1,0,3}$, 
				$\mathcal{I}_{4,0,2}$, 
				$\mathcal{I}_{4,0,3}$, 
				$\mathcal{I}_{1,4}$, and
				$\mathcal{I}_{2,3}$ 
				be their corresponding monotonicity intervals.
				
				Let $\mathcal{I}_1 = \mathcal{I}_{1,4}
                                \cap \mathcal{I}_{2,3}$, and let
                                $\mathcal{I}_2 = \mathcal{I}_{1,0,2}
                                \cap \mathcal{I}_{1,0,3} \cap
                                \mathcal{I}_{4,0,2} \cap
                                \mathcal{I}_{4,0,3}$.  If
                                $\mathcal{I}_1 \ne \emptyset$ and
                                $\mathcal{I}_2 \ne \emptyset$, then
                                $T$ is $\mathcal{D}$-monotone for
                                every set $\mathcal{D}=\{d_1,d_2\}$
                                with $d_1 \in \mathcal{I}_1$ and
                                $d_2 \in \mathcal{I}_2$.  Otherwise,
                                if $\mathcal{I}_1$ or $\mathcal{I}_2$
                                is empty, then $T$ is not 2-directional
                                monotone.
			\end{enumerate}
			
			\begin{figure}[tb]
				\centering
				\begin{subfigure}{.32\linewidth}
					\centering
					\includegraphics[page=1]{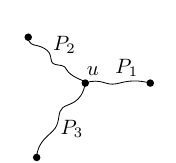}
					\subcaption{}
					\label{fig:rec_case2a}
				\end{subfigure}
				\begin{subfigure}{.32\linewidth}
					\centering
					\includegraphics[page=2]{recognition}
					\subcaption{}
					\label{fig:rec_case2b}
				\end{subfigure}
				\begin{subfigure}{.32\linewidth}
					\centering
					\includegraphics[page=3]{recognition}
					\subcaption{}
					\label{fig:rec_case3}
				\end{subfigure}
				\caption{Tree configurations for monotone tree recognition for the case that  tree $T$ is not a path. (a)  $T$ contains  a single degree-3 vertex. (b) $T$ contains  exactly two degree-3 vertices. (c) $T$ contains  a single degree-4 vertex.}
				\label{fig:rec_lemma}
			\end{figure}
			
                        \item {\it $\Delta(T)=4$ (\cref{fig:rec_case3}):}
                        Then, $T$ consists of a vertex $u$ of degree~4
                        and four paths $P_1, P_2, P_3$, and $P_4$
                        that have $u$ as their common endpoint and
                        that lead to distinct leaves.
			
			For $T$ to be $\mathcal{D}$-monotone, all six paths 
			$P_1 P_2$, 
			$P_1 P_3$,
			$P_1 P_4$,
			$P_4 P_2$,
			$P_4 P_3$, and
			$P_2 P_3$ must be monotone.
                        If this is the case, let  
			$\mathcal{I}_{1,2}$,
			$\mathcal{I}_{1,3}$,
			$\mathcal{I}_{1,4}$,
			$\mathcal{I}_{4,2}$,
			$\mathcal{I}_{4,3}$ and 
			$\mathcal{I}_{2,3}$ be their corresponding monotonicity intervals.
			Let
                        $\mathcal{I}_{1} = \mathcal{I}_{1,4} \cap
                        \mathcal{I}_{2,3} \cap \mathcal{I}_{1,3} \cap
                        \mathcal{I}_{4,2}$, and let
                        $\mathcal{I}_{2} = \mathcal{I}_{1,2} \cap
                        \mathcal{I}_{4,3} \cap \mathcal{I}_{1,3} \cap
                        \mathcal{I}_{4,2}$.  If
                        $\mathcal{I}_{1} \neq \emptyset$ and
                        $\mathcal{I}_{2} \neq \emptyset$, then $T$ is
                        $\mathcal{D}$-monotone for every set
                        $\mathcal{D}=\{d_1,d_2\}$ with
                        $d_1 \in \mathcal{I}_1$ and
                        $d_2 \in \mathcal{I}_2$.  Otherwise, if
                        $\mathcal{I}_1$ or $\mathcal{I}_2$ is empty,
                        then $T$ is not 2-directional monotone.
		\end{enumerate}
		Note that all checks for path monotonicity and the computation of all
		monotonicity intervals can be accomplished in linear time.
	\end{proof}
	
	\section{Minimum $k$-Directional Monotone Spanning Trees}\label{se:k-monotone}
	
	In this section we prove that the $\MMST(S,k)$ problem is in XP when
	parameterized by $k$.  To this end, we first describe an XP algorithm
	for the $\MMST(S,\cal D)$ problem, parameterized by $k=|\cal D|$.

	A \emph{homeomorphically
		irreducible tree (HIT)} is an embedded tree without vertices of
	degree two~\cite{10.1007/BF02559543}.
	In this context, two trees $T_1$ and $T_2$ have the same \emph{topology} if they are
	(possibly different) subdivisions of the same HIT. 
	Given a HIT $H$ and any embedded tree $T$ that is a
        subdivision of $H$, we say that $H$ \emph{corresponds to}~$T$.
	Since for a vertex of degree two the circular order of its incident edges is unique, the embedding of a tree uniquely defines the embedding of the corresponding HIT.

Let $S$ be a set of $n$ points, and let ${\cal D}$ be a set of $k$
(pairwise non-opposite) distinct directions. In a
$\mathcal{D}$-monotone spanning tree~$T$, a \emph{branching vertex} is
a vertex of degree at least~3 and a \emph{leaf path} is a path of
degree-2 vertices from a branching vertex to a leaf. A \emph{branch}
$B_{u,v}$ is a path that connects two adjacent branching vertices~$u$
and~$v$ via a sequence of degree-2 vertices. Both a leaf path and a
branch may consist of a single edge.
Note that, given a tree $T$ and the corresponding HIT $H$, an internal
vertex of $H$ corresponds to a branching vertex of $T$, a leaf of $H$
corresponds to a leaf path of $T$, and an edge between two internal
vertices of~$H$ corresponds to a branch of~$T$.
Further, for any pair of vertices $u$ and $v$ (not necessarily  adjacent)  let $T_{u,v}$  be the  subtree  of $T$ consisting of vertex $u$ and all subtrees hanging from $u$ except the one containing vertex $v$.  
We recall that $k$ directions define $2k$ wedges and that a $\D$-monotone spanning tree $T$ can have at most $2k$ leaves (\cref{pr:2k-leaves}).

Consider a directed \D-monotone path
$P=\langle v, v_1, \dots, v_{r} \rangle$ originating at vertex
$v$ and let $\sec(P)$ be its sector of monotonicity as defined in
\cref{se:properties-paths-trees}. Let
$\Wutil_{P}=\{W_1, \ldots, W_m\}$ be the smallest set of consecutive
wedges (in CCW order) whose union contains $\sec(P)$.  Note that, due
to \D-monotonicity, $|\Wutil_{P}| \le k$.  If $\overleftarrow{P}$ is
the reverse path of $P$, that is,
$\overleftarrow{P}=\langle v_{r}, \dots, v_1, v \rangle$, then
$\Wutil_{\overleftarrow{P}}$ consists of the wedges opposite to those
in~$\Wutil_{P}$.
Refer to \cref{fig:assigned_used_wedges} where
$\Wutil_{P}=\{W_1, \ldots, W_5\}$.  We say that path $P$
\emph{utilizes} wedge set $\Wutil_{P}$.

We now present an overview of the algorithm for solving
the $\MMST(S, \D)$ problem. Our algorithm examines every HIT with at
most $2k$ leaves.  In \cref{sec:bound_on_HITS}, we provide an
$O(7^{2k} \cdot (2k)!)$ bound on the number of HITs supported by a constructive enumeration scheme.
Since there are many ($\D$-monotone)
spanning trees that are subdivisions of the same HIT, we examine for
each HIT all of its $\D$-monotone spanning trees on $S$.  Let
$H$ be the HIT under consideration,
and let $\ell$ and $b$ be the numbers of leaves and branching vertices
of $H$, respectively.  Let $M$ be one of the $O(n^b)$ possible
mappings of the $b$ branching vertices to points in~$S$.  Let $A$ be
an assignment of the wedges of $\W_\D$ to the leaves of $H$ so that
each leaf receives a distinct set of consecutive wedges.  Assigning
(as part of~$A$) the set of consecutive wedges $\W^A$ to a
leaf $\lambda$ incident to a branching vertex $v$ of $H$ can be
interpreted as our intention to cover all points in region $\W^A(v)$
by the monotone leaf path~$P$ that ends at the leaf~$\lambda$.
As shown in \cref{fig:assigned_used_wedges}, the monotone leaf path may utilize a
set of consecutive wedges $\Wutil_{P} \subseteq \W^A$, i.e., some of
the leading and/or trailing wedges of $\W^A$ may not be utilized by~$P$.

The point set~$S$, the set $\D$ of $k$ (pairwise non-opposite)
directions, the HIT $H$, together with mapping $M$ and
assignment~$A$, form an instance of a restricted problem that asks for
a minimum $\D$-monotone spanning tree that has $H$ as its HIT and
respects mapping~$M$ and assignment~$A$.  We denote this problem
instance as \emph{$\MMST(S, \D, H, M, A)$}.  Note that such a
monotone spanning tree may not exist. If it exists, it turns out that it is unique. The algorithm for solving the instance
$\MMST(S, \D, H, M, A)$ is repeatedly used by the XP algorithm that
we describe in \cref{sub:xp}.

	\begin{figure}[tb]
		\centering
		\begin{subfigure}{.45\linewidth}
			\centering
			\includegraphics[page=1]{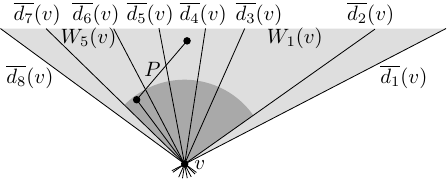}
			
			\label{fig:assigned_used_wedges1}
		\end{subfigure}
		\caption{A leaf path $P$ which is assigned seven wedges (light gray sector) but only uses five of them (shaded gray): It 
			is not monotone with respect to $d_3, d_4, d_5, d_6$.}
		\label{fig:assigned_used_wedges}
	\end{figure}
	
\subsection{A characterization of $\D$-monotone spanning trees}
\label{se:k-properties}
Let $\Wutil_P \subseteq \W_\D$ be the set of wedges utilized by an
oriented path $P$ (either a leaf path or a branch) originating from
vertex $v$ and let $\Wutil_P(v)$ be the region of the plane determined
by $\Wutil_P$ translated such that $v$ is its apex.  Moreover, let
$\Wutil_{u,v} \subseteq \W_\D$ be the smallest set of consecutive
wedges that contains all wedges utilized by either leaf paths or
branches oriented away from $u$ in~$T_{u,v}$ and does not contain the
wedge utilized by the edge out of $u$ that leads to vertex $v$. For an
example refer to \cref{fig:Wuv_wedgeSet}. Note that for the case that $u$ and/or $v$ is a leaf, $\Wutil_{u,v}= \emptyset$ and/or $\Wutil_{v,u}= \emptyset$. Let $\Wutil_{u,v}(u)$ be the
region defined by wedges in $\Wutil_{u,v}$ translated such that~$u$ is
their apex. 
An immediate consequence of \cref{lemma:monotone-path-range} is the following corollary.

	\begin{figure}[tb]
	\begin{subfigure}{.55\textwidth}
		\includegraphics[page=1]{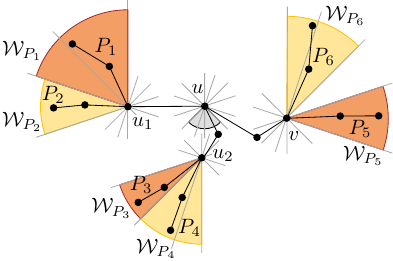}
		\subcaption{}
		\label{fig:Wuv_wedgeSet1}
	\end{subfigure}
	\hfill
	\begin{subfigure}{.43\textwidth}
		\includegraphics[page=2]{wedges_util}
		\subcaption{}
		\label{fig:Wuv_wedgeSet2}
	\end{subfigure}
	\caption{(a) A monotone tree and its sets of utilized wedges
          for each leaf paths, where the center of the circle of directions coincides with its branching vertices.
          (b) All sets of utilized wedges drawn on the same unit
          circle. Set $\Wutil_{u,v}$ (resp. $\Wutil_{v,u}$) consists of all wedges in the
          blue (resp. gray)~region.}
	\label{fig:Wuv_wedgeSet}
\end{figure}

\begin{corollary}\label{le:monot_external_dir}
  Let $S$ be a set of points,
  let~$\mathcal{D}=\{d_1, d_2, \dots, d_k\}$ be a set of $k$ (pairwise
  non-opposite) directions such that $S$ is in ${\cal D}$-general
  position and let~$T$ be a $\cal D$-monotone spanning tree
  of~$S$. Let $P$ be a path originating at vertex $u$ of
  $T$. Given a direction $d \in \D$, $P$ is $d$-monotone if and only
  if $\overline{d}(u)$ does not intersect the interior of
  $\Wutil_P(u)$.
\end{corollary}

Let $B_{u,v}$ be a branch of a monotone tree $T$ connecting branching
vertices $u$ and $v$. Recall that, due to monotonicity of $B_{u,v}$,
we have that $|\Wutil_{B_{u,v}}| \le k$. Let
$R_{u,v} = \Wutil_{B_{u,v}}(u) \cap \Wutil_{B_{v,u}}(v)$. If
$|\Wutil_{B_{u,v}}|<k$ then $R_{u,v}$ is a parallelogram (refer to
\cref{fig:branching_region1}). Otherwise, if $|\Wutil_{B_{u,v}}|=k$ then $R_{u,v}$ is a
strip bounded by parallel lines $\overline{d}(u)$ and
$\overline{d}(v)$ where $d$ is the direction of monotonicity of
$B_{u,v}$ (refer to \cref{fig:branching_region2}). We refer to $R_{u,v}$ as the
\emph{region of branch} $B_{u,v}$. Similarly, if $P_{u, \lambda}$ is
a leaf path from branching  vertex $u$ to leaf $\lambda$ then we
define the \emph{region of the leaf path}
$R_{u, \lambda}=\Wutil_{P_{u, \lambda}}(u) \cap
\Wutil_{P_{\lambda, u}}(\lambda)$.
     
     \begin{figure}[tb]
     	\centering
     	\begin{subfigure}{.48\linewidth}
     		\centering
     		\includegraphics[page=1]{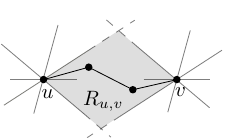}
     		\subcaption{}
     		\label{fig:branching_region1}
     	\end{subfigure}
     	\begin{subfigure}{.48\linewidth}
     		\centering
     		\includegraphics[page=2]{branching_region}
     		\subcaption{}
     		\label{fig:branching_region2}
     	\end{subfigure}
     	\caption{(a) If $|\Wutil_{B_{u,v}}|<k$ then $R_{u,v}$ is a parallelogram. (b) If f $|\Wutil_{B_{u,v}}|=k$ then $R_{u,v}$ is an unbounded strip.}
     	\label{fig:branching_region}
     \end{figure}

	\begin{lemma}
  \label{le:k-directional-properties}
  Let $S$ be a set of points,
  let~$\mathcal{D}=\{d_1, d_2, \dots, d_k\}$ be a set of $k$ (pairwise
  non-opposite) directions such that $S$ is in ${\cal D}$-general
  position, and let~$T$ be a $\cal D$-monotone spanning tree of~$S$.
  Then~$T$ has the following properties.
  \begin{enumerate}[({I}1)]
  \item \label{enum:path_subtree_inArea1}
    Let $P$ be a path originating at vertex $u$ of $T$.  Then, $P$
    lies in $\Wutil_{P}(u)$.
  \item \label{enum:path_subtree_inArea2}
    Let $u,~v$ be two vertices of $T$. Then, $T_{u,v}$ lies in
    $\Wutil_{u,v}(u)$.
  \item \label{enum:disjoint_path_util_area}
    Let $P_1$ and $P_2$ be two edge-disjoint paths originating at
    internal vertices $u$ and $v$ of $T$ and terminating at leaves of
    $T$. Then, sets $\Wutil_{P_1}$ and $\Wutil_{P_2}$ are disjoint
    and regions $\Wutil_{P_1}(u)$ and $\Wutil_{P_2}(v)$ are disjoint.
  \item \label{enum:disjoint_Wuv_wedge_area}
    Consider any two vertices $u,~v$ of $T$. Then, sets
    $\Wutil_{u,v}$ and $\Wutil_{v,u}$ are disjoint and regions
    $\Wutil_{u,v}(u)$ and $\Wutil_{v,u}(v)$ are disjoint.
  \item \label{enum:branch-disjoint}
    Let $P_{u,v}$ be either a branch or a leaf path of $T$.  Then,
    $\Wutil_{P_{u,v}} \cap \Wutil_{u,v}= \emptyset$ and
    $R_{u,v} \cap \Wutil_{u,v}(u)= \emptyset$.
  \end{enumerate}
  \end{lemma}

	\begin{proof}
		We show the properties stated in the lemma one by one.
		\begin{enumerate}[({I}1)]
			\item Let $d_1$ and $d_2$ be the two directions orthogonal to the
			boundaries of $\Wutil_{P}$.  Then, due to
			\cref{le:monot_external_dir},
			path $P$ is both $d_1$- and
			$d_2$-monotone.  Let $u_1$ be the vertex incident to $u$ on
			$P$. Then, by definition of $\Wutil_{P}$, we have that $u_1$ lies in
			$\Wutil_{P}(u)$.  Now let, for the sake of contradiction, $x$ be a
			vertex of $P$ that lies outside the region
			$\Wutil_{P}(u)$. Observe that vertices $x$ and $u$ lie in the same
			halfplane with respect either to $\overline{d_1}(u_1)$ or
			$\overline{d_2}(u_1)$. Therefore, due to
			\cref{pr:monotone-half-plane}, path $P$ is not monotone with
			respect to both $d_1$ and $d_2$. A clear contradiction.
			
			\item Consider first a path $P$ from vertex $u$ oriented towards an
			arbitrary leaf $\lambda$ of $T_{u,v}$.  By~(I\ref{enum:path_subtree_inArea1}), we have that path $P$
			lies in region $\Wutil_P(u)$.  Since path $P$ is composed of
			branches (zero or more) and a single leaf path in $T_{u,v}$, it
			follows that $\Wutil_P \subseteq \Wutil_{u,v}$ and, in turn, that
			path $P$ lies in region $\Wutil_{u,v}(u)$. Since the union of all
			paths from $u$ to the leaves of $T_{u,v}$ covers all branches and
			leaf paths in $T_{u,v}$ is follows that $T_{u,v}$ lies in
			$\Wutil_{u,v}(u)$.
			
			\item Let $P_1=T(u, {\lambda})$ and
			$P_2=T(v, \mu)$ where $u$ and $v$ are
			internal vertices of $T$ and $\lambda$ and $\mu$
			are the corresponding leaves. Since $P_1$ and $P_2$ are
			edge-disjoint, path $P=T({\lambda}, \mu)$ from
			${\lambda}$ to $\mu$ is composed of
			$\overleftarrow{P_1}$, $T(u,v)$ and $P_2$. Since $T$ is
			$\D$-monotone, $P$ must be $\D$-monotone with respect to at least
			one direction, say $d \in \D$. It follows that for any internal
			vertex $w$ in the path the oriented subpaths $T({\lambda}, w)$
			and $T(w, \mu)$ lie in different halfplanes with respect
			to $\overline{d}(w)$.
			\begin{enumerate}[(a)]
				\item For the sake of contradiction assume that $\Wutil_{P_1}$ and
				$\Wutil_{P_2}$ overlap. Then, for $w=u$ we get that all vertices
				of $P_1$ must lie behind $\overline{d}(u)$. At the same time,
				for $w=v$ we get that all vertices of path $P_2$ must lie ahead
				of $\overline{d}(v)$. However, due to the fact that
				$\Wutil_{P_1}$ and $\Wutil_{P_2}$ overlap, no such direction $d$
				exists; a contradiction to the monotonicity of path $P$.
				\item As shown in (a), subpaths
				$\overleftarrow{P_1}=T(\lambda, u)$ and
				$P_2=T(v, \mu)$ lie in different halfplanes with
				respect to $\overline{d}(u)$ (or $\overline{d}(v)$). Given that
				$\Wutil_{P_1}$ and $\Wutil_{P_2}$ are disjoint, we conclude that
				$\Wutil_{P_1}(u) ~\cap~ \Wutil_{P_2}(v) = \emptyset$.
			\end{enumerate}
			
			\item Observe that if either one of the vertices $u$ or $v$, say
			$v$, is a leaf, then
			$\Wutil_{u,v}=\emptyset$. Therefore, both statements in the lemma
			trivially hold. The same applies in the case where both $u$ and
			$v$ are distinct leaves. So, in the rest of the proof we assume
			that $u$ and $v$ are any two internal tree vertices.
			\begin{enumerate}[(a)]
				\item For the sake of contradiction assume that $\Wutil_{u,v}$ and
				$\Wutil_{v,u}$ overlap.  By~(I\ref{enum:disjoint_path_util_area}), we know that
				there do not exist paths $P_1$ and $P_2$ terminating at leaves
				of $T$ that belong to $T_{u,v}$ and $T_{v,u}$, respectively,
				such that $\Wutil_{P_1}$ overlaps with $\Wutil_{P_2}$. So,
				without loss of generality, we assume that there exists a path
				$P=T(w, \lambda)$ terminating at leaf $\lambda$
				in $T_{v,u}$ such that $\Wutil_{P} \subset
				\Wutil_{u,v}$. Furthermore, let $P_1$ and $P_2$ be the paths
				terminating at leaves of $T_{u,v}$ utilizing the leading and
				trailing wedges of $ \Wutil_{u,v}$, respectively.  Refer to
				\cref{fig:disjoint_Wuv}.  Since $T$ is an embedded monotone
				tree, so is its subtree $T'$ that consists of $T_{u,v}$ and the
				path from $u$ to $\lambda$ (which passes from $v$ and $w$) and
				uses the same embedding as $T$. Consider path $P'$ from $u$ to
				$\lambda$ and let $(u,u')$ be its first edge. Then, $P'$ is a
				path terminating at leaf $\lambda$ and its set of utilized
				wedges $\Wutil_{P'}$ includes the wedge utilized by edge
				$(u,u')$, the wedges in $\Wutil_{P}$ and the wedges in
				$\Wutil_{T(u', w)}$. Thus, path $P'$ is, at least, utilizing all
				wedges also utilized by either $P_1$ or $P_2$. Without loss of
				generality, assume that $\Wutil_{P'}$ intersects with
				$\Wutil_{P_1}$. However, given that both $P'$ and $P_1$ are
				edge-disjoint paths terminating at leaves, by~(I\ref{enum:disjoint_path_util_area}), we have that
				$\Wutil_{P'}$ and $\Wutil_{P_1}$ are disjoint, a clear
				contradiction. We conclude that $\Wutil_{u,v}$ and
				$\Wutil_{v,u}$ are disjoint.
				\item Recall that, due to (a), we have that $\Wutil_{u,v}$ and
				$\Wutil_{v,u}$ are disjoint. By the definition of $\Wutil_{u,v}$
				and $\Wutil_{v,u}$ it follows that the leading and the trailing
				wedges of $\Wutil_{u,v}$ and $\Wutil_{v,u}$ are utilized. Let
				$W_1$ and $W_1'$ be the leading and the trailing wedges of
				$\Wutil_{u,v}$ and let $W_2$ and $W_2'$ be the leading and the
				trailing wedges of $\Wutil_{v,u}$ in counterclockwise
				order. Assume for the sake of contradiction that areas
				$\Wutil_{u,v}(u)$ and $\Wutil_{v,u}(v)$ intersect. Then, at
				least one of $W_1(u), W_1'(u)$ intersects with at least one of
				$W_2(v), W_2'(v)$. W.l.o.g., let $W_1(u)$ intersect with
				$W_2(v)$ and let $e_1=(u_1, u_2)$ be the oriented edge away from
				$u$ that utilizes wedge $W_1$ and $e_2=(v_1, v_2)$ be the
				oriented edge away from $v$ that utilizes wedge $W_2$.  Let
				$P_1=T(u, u_{2})$ be the path of $T$ originating at $u$ that
				utilizes wedge $W_1$ and let $P_2=T(v, v_{2})$ be the path of
				$T$ originating at $v$ that utilizes wedge $W_2$.  Paths $P_1$
				and $P_2$ are edge-disjoint and terminate at leaves of $T$.  Due
				to~(I\ref{enum:disjoint_path_util_area}), regions
				$\Wutil_{P_1}(u)$ and $\Wutil_{P_2}(v)$ are disjoint. A
				contradiction.
			\end{enumerate}
			
			\begin{figure}[tb]
				\begin{subfigure}{.55\textwidth}
					\includegraphics[page=1]{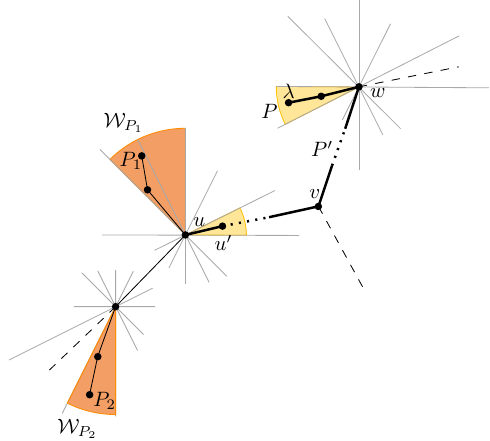}
					\subcaption{}
					\label{fig:disjoint_Wuv1}
				\end{subfigure}
				\hfill
				\begin{subfigure}{.43\textwidth}
					\includegraphics[page=2]{wedges_distinct}
					\subcaption{}
					\label{fig:disjoint_Wuv2}
				\end{subfigure}
				\caption{Monotone embedded tree used in the proof of (I\ref{enum:disjoint_Wuv_wedge_area}).}
				\label{fig:disjoint_Wuv}
			\end{figure}
			
			\item Let $T'$ be the subtree of $T$ formed by $T_{u,v}$ and
			$P_{u,v}$. Tree $T'$ is $\D$-monotone since it is a subtree of
			$T$. Consider first the case where $P_{u,v}$ consists only of edge
			$(u,v)$. Then, by definition, edge $(u,v)$ utilizes a wedge which
			is not contained in $\Wutil_{u,v}$ and, thus, it immediately
			follows that $\Wutil_{P_{u,v}} \cap \Wutil_{u,v}= \emptyset$.
			Consider now the case where path $P_{u,v}$ contains at least one
			intermediate vertex. Let $(u,w)$ be the edge of $P_{u,v}$ incident
			to $u$. Edge $(u,w)$ utilizes a wedge which is not contained in
			$\Wutil_{u,v}$. Consider now vertices $u$ and $w$ of $T'$ and
			subtrees $T'_{u,w}$ and $T'_{w,u}$.  By~(I\ref{enum:disjoint_Wuv_wedge_area}), we have that
			$\Wutil_{u, w}$ and $\Wutil_{w,u}$ are disjoint with respect to
			$T'$ and, therefore, they are also disjoint with respect to $T$.
			Since $T'_{u,v}=T'_{u,w}$ and $P_{u,v}$ is composed of edge
			$(u,w)$ and $T'_{w, u}$, we conclude that
			$\Wutil_{P_{u,v}} \cap \Wutil_{u,v}= \emptyset$.
			
			Now, observe that $R_{u,v} \subset \Wutil_{P_{u,v}}(u)$. Since
			regions $R_{u,v}$ and $\Wutil_{u,v}(u)$ have the same apex and are
			contained in the disjoint sets of utilized wedges
			$\Wutil_{P_{u,v}}$ and $\Wutil_{u,v}$, respectively, we also
			conclude that $R_{u,v} \cap \Wutil_{u,v}(u)= \emptyset$.\qedhere
		\end{enumerate}
	\end{proof}
	
	The following theorem characterizes $\D$-monotone trees.
	
\begin{theorem}
  \label{lem:Branches_large_Wuv}
  Let $S$ be a set of points,
  let~$\mathcal{D}=\{d_1, d_2, \dots, d_k\}$ be a set of $k$ (pairwise
  non-opposite) directions such that $S$ is in ${\D}$-general
  position, and let~$T$ be a spanning tree of~$S$.  Then $T$ is
  $\D$-monotone if and only if:
		\begin{enumerate}[(a)]
			\item Every leaf path and every branch $P$ in $T$ is $\D$-monotone.
			\item For every two leaf paths $P_1$ and $P_2$ incident to branching vertices $u$ and $v$, respectively, $\Wutil_{P_1}$ and $\Wutil_{P_2}$ are disjoint.
			\item Let $P_{u,v}$ be either a branch or a leaf path of $T$. Then $\Wutil_{u,v}(u) \cap R_{u,v} = \emptyset$. 
		\end{enumerate}
	\end{theorem}
	\begin{proof}
          $(\Rightarrow)$ Since $T$ is a $\D$-monotone tree, any
          subtree of $T$ is also $\D$-monotone and thus (a)
          holds. Statements~(b) and~(c) follow from~(I\ref{enum:disjoint_path_util_area})
          and~(I\ref{enum:branch-disjoint}), respectively.
		
          \noindent $(\Leftarrow)$ For the monotonicity of $T$, it
          suffices to show that for any two leaves
          $\lambda, \mu$ there is a $\D$-monotone path between them. Let $P_{u,
            \lambda}$ and $P_{v, \mu}$ be the leaf paths to
          $\lambda$ and $\mu$ where $u$ and $v$ are the
          branching vertices they are incident to,
          respectively. First,
		consider the case where $P_{u, \lambda}$ and $P_{v, \mu}$ are incident to the same
		vertex, that is, $u=v$. Due to (a), both leaf paths
		are $\D$-monotone and, thus, $|\Wutil_{P_{u,
				\lambda}}|\le k$ and $|\Wutil_{P_{v, \mu}}|\le
		k$. Moreover, due to (b), it follows that
		$\Wutil_{P_{u, \lambda}}$ and $\Wutil_{P_{v, \mu}}$
		are disjoint. Therefore, there exists a direction~$d$
		in~\D such that $\overline{d}(u)$
		separates $\Wutil_{P_{u, \lambda}}(u)$ and $\Wutil_{P_{{v, \mu}}}(v)$ and does not intersect the
		interior of either of them.  Due to
		\cref{le:monot_external_dir},
		$P_{u, \lambda}$ and $P_{v, \mu}$ are both
		$d$-monotone and, additionally, they lie on different halfplanes
		with respect to $\overline{d}(u)$.  Therefore, the
		path from~$\lambda$ to~$\mu$ is $d$-monotone, and thus
		\D-monotone.
		
		Consider now the case where $u \ne v$. Let $\mathcal{B}=\{u=b_1, \dots, b_r=v, b_{r+1}=\mu\}$ be the sequence of the
		branching vertices on $T(\lambda, \mu)$ in order of appearance, where, for convenience, we treat leaf $\mu$ as a branching vertex. 
		Due to
		\cref{le:monot_external_dir}, 
		it suffices
		to show that there is a
		direction $d$ such that line $\overline{d}(\mu)$  does not intersect the interior of~$\Wutil_{P_{\mu, \lambda}}(\mu)$. 
		
		Let $\mathcal{P}_i=P_{b_i, \lambda}$ denote the subpath of
		$P_{\mu, \lambda}$ from vertex $b_i$ to leaf $\lambda$. We show by
		induction on the size of $\mathcal{B}$ that \emph{for every
			$i \in \{1,\dots,r+1\}$, $|\Wutil_{\mathcal{P}_i}|\le k$}. Since
		$\mathcal{P}_{r+1}$ is by definition the oriented path from $\mu$ to
		$\lambda$, the fact that $|\Wutil_{\mathcal{P}_{r+1}}|\le k$ together
		with 
		\cref{le:monot_external_dir}
		guarantee that there exists a
		direction $d \in \D$ such that the path from $\lambda$ to $\mu$ is
		$d$-monotone. For the basis of the induction observe that
		$\mathcal{P}_1$ is the leaf path $P_{u, \lambda}$, which is
		$\D$-monotone by (a). For the induction hypothesis, assume that
		$|\Wutil_{\mathcal{P}_i}| \le k$ for $i \le m$. We show that
		$|\Wutil_{\mathcal{P}_{m+1}}| \le k$. Assume, for a contradiction, that
		$|\Wutil_{\mathcal{P}_{m+1}}| >k$. Since $\mathcal{P}_{m+1}$ consists
		of $\mathcal{P}_m$ and of the branch $B_{b_m, b_{m+1}}$, the wedges of
		$\Wutil_{\mathcal{P}_{m+1}} \setminus \Wutil_{\mathcal{P}_{m}}$ are
		due to branch $B_{b_m, b_{m+1}}$. Let $W^m_1$ and $W^m_2$ be the
		leading and the trailing wedges (in CCW order) of
		$\Wutil_{\mathcal{P}_{m}}$ and let $W^{m+1}_1$ and $W^{m+1}_2$ be the
		leading and the trailing wedges (in CCW order) of
		$\Wutil_{\mathcal{P}_{m+1}}$. Observe first that either
		$W^m_1= W^{m+1}_1$ or $W^m_2= W^{m+1}_2$ (refer to
		\cref{fig:wedges_smaller_k1}). If this was not the case, then
		$|\Wutil_{B_{b_{m+1}, b_{m}}}|>k$ which contradicts the fact that all
		branches are $\D$-monotone.
		
		\begin{figure}[tb]
			\centering
			\begin{subfigure}{.48\linewidth}
				\centering \includegraphics[page=1]{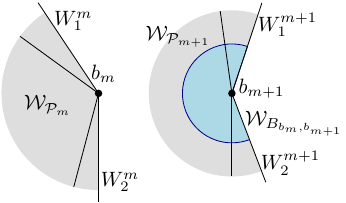}
				\subcaption{}
				\label{fig:wedges_smaller_k1}
			\end{subfigure}
			\begin{subfigure}{.48\linewidth}
				\centering \includegraphics[page=2]{wedges_smaller_k}
				\subcaption{}
				\label{fig:wedges_smaller_k2}
			\end{subfigure}
			\caption{Different cases examined in the proof of
				\cref{lem:Branches_large_Wuv}}
			\label{fig:wedges_smaller_k}
		\end{figure}

		Now assume without loss of generality, that $W^m_2= W^{m+1}_2$ (refer to \cref{fig:wedges_smaller_k2}). Then the leading wedge of $\Wutil_{B_{b_{m+1}, b_{m}}}$ is $W^{m+1}_1$ and branch $B_{b_{m+1}, b_{m}}$ utilizes at most $k$ wedges since it is $\D$-monotone. Moreover, $\Wutil_{B_{b_{m+1}, b_{m}}}(b_{m+1})$ contains vertex $b_{m}$, since otherwise it would not be $\D$-monotone. Consider now the utilized wedge set $\Wutil_{B_{b_{m}, b_{m+1}}}$ of branch  $B_{b_{m}, b_{m+1}}$ which consists of the opposite wedges of $\Wutil_{B_{b_{m+1}, b_{m}}}$. Then its leading wedge is the opposite of $W^{m+1}_1$ and is located before wedge $W^m_2$ (in CCW order) and its trailing wedge is located after $W^m_2$ (in CCW order). Thus $\Wutil_{\mathcal{P}_i}(b_i)$ intersects the region of the branch (refer to the green parallelogram in \cref{fig:wedges_smaller_k2}). This is a contradiction, since due~to~(c),~$\Wutil_{\mathcal{P}_i}(b_m) \subset \Wutil_{b_m, b_{m+1}}(b_m)$. Note that considering $\mu$ as a branching vertex does not affect the correctness of the proof.
	\end{proof}

	\subsection{An upper bound on the number of HITs}
	\label{sec:bound_on_HITS}
	
	In this section, we prove an upper bound on the number $n_{\ell}$ of
	HITs with at most $\ell$ leaves.  We remark that an upper bound
        can be derived from a result of Harary, Robinson, and
	Schwenk~\cite{harary_robinson_schwenk_1975}.  However, their
        result does not yield an algorithm to generate all
	different HITs with at most $\ell$ leaves. For this reason, we
	give a weaker upper bound that is based on a
	generation scheme. Note that our generation scheme can
	generate the same HIT multiple times.
	
	\begin{lemma}
		\label{le:numEmbeddedTrees}
		The number of different HITs with at most $\ell$ leaves is
		$O(7^\ell \cdot \ell!)$, and these HITs can be enumerated in
		$O(7^\ell \cdot \ell!)$ time.
	\end{lemma}

	\begin{proof}
		Denote by $\bar{n}_i$ the number of different HITs
		with \emph{exactly} $i$ leaves, for $i \geq 2$. 
		We now prove, by induction on $i$, that
		$\bar{n}_i \leq 7^i \cdot i!$ for $i \geq 2$. For
		$i=2$, there exists only one possible HIT,
		i.e., the tree consisting of a single edge. Suppose that $i>2$. A
		HIT with $i$ leaves can be obtained from a HIT
		with $i-1$ leaves by means of one of two operations: by attaching
		an edge (and a leaf) to an internal vertex (we call this
		Operation~1) or by subdividing an edge and attaching a new edge to
		the degree-two vertex created by the subdivision (we call this
		Operation~2); see \cref{fi:HIT}.
		
                \begin{figure}[tb]
                  \centering
                  \includegraphics{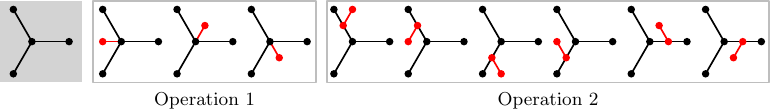}
                  
                  \smallskip
                  
                  \caption{A HIT with three leaves and all nine HITs
                    that can be generated from it by means of
                    Operations~1 and~2.}
                  \label{fi:HIT}
                \end{figure}
		
		Let $T$ be a HIT with $i-1$ leaves.  Given~$T$, let $V$
		be the set of vertices, let $I$ be the set of internal vertices, let
		$L$ be the set of leaves, and let $m$ be the number of edges of~$T$.
		If we perform Operation~1 on an internal vertex $v$ of the tree $T$,
		we can obtain $\deg(v)$ different HITs, which have the same topology but different embedding depending on the position of the new edge in the circular order
		around~$v$.  Thus, the number of different HITs that can be
		generated starting from $T$ by performing Operation~1 is
		$\sum_{v \in I}\deg(v)$. We have
		$\sum_{v \in V}\deg(v)=\sum_{v \in I}\deg(v) + \sum_{v \in
			L}\deg(v)=2m$. The term $\sum_{v \in L}\deg(v)$ is equal to the
		number of leaves, that is, $i-1$; moreover, since the number of
		leaves is $i-1$, the number of vertices is at most $2i-4$, and
		the number of edges $m$ is $2i-5$. Thus, we obtain
		$\sum_{v \in I}\deg(v)=2(2i-5)-i+1=3i-9$. If we perform
		Operation~2 on an edge $e$ of $T$, we can obtain $2$ different
		HITs depending on the side of $e$ where the new edge is
		added. Thus, from the tree $T$ we can obtain at most
		$\sum_{v \in I}\deg(v)+2m=3i-9+4i-10 = 7i-19$ different
		HITs, which implies
		$\bar{n}_i \leq 7i \bar{n}_{i-1}$. By induction,
		$\bar{n}_{i-1} \leq 7^{i-1}(i-1)!$ and therefore
		$\bar{n}_i \leq 7^i \cdot i!$. 
		
		The number of different
		HITs with \emph{at most} $\ell$ leaves can now be computed as
		$\sum_{i=2}^{\ell}\bar{n}_i \leq \sum_{i=2}^{\ell} 7^i \cdot i! \leq
		\ell! \sum_{i=2}^{\ell}7^i \leq \frac{7^\ell\cdot \ell!}{14}$.
		Clearly, all these HITs can be generated starting from the
		single tree with two leaves as described above by performing
		Operations~1 and~2. Since each operation can be executed in $O(1)$
		time, the whole set can be generated in $O(7^\ell \cdot \ell!)$
		time.
	\end{proof}	
	
	\subsection{The XP algorithm}
	\label{sub:xp}

We now present an XP algorithm for solving the $\MMST(S, \D )$ problem. Toward this end, we start by showing that an instance of the 
$\MMST(S, \D, H, M, A)$ problem can be solved in polynomial time with respect to  the sizes of $S$ and $\D$. The algorithm that solves the $\MMST(S, \D, H, M, A)$ problem is then repeatedly applied by our XP algorithm for the $\MMST(S, \D )$ problem.
        
        \begin{lemma}
          \label{le:mmst_sdtma}
          Let $S$ be a set of $n$ points, let~\D be a set of $k$
          (pairwise non-opposite) directions, let~$H$ be a HIT, let
          $M$ be a mapping of the internal vertices of $H$ to points
          of $S$, and let $A$ be an assignment of $\W_\D$ to the
          leaves of $H$ so that each leaf receives a distinct set of
          consecutive wedges.  Then the 
          $\MMST(S, \D, H, M, A)$ problem can be solved in $O(n\log n+nk+k)$
          time.
        \end{lemma}

	\begin{proof}
		Let $\mathcal{B}^{H}$ be the internal vertices of $H$. As discussed, every internal vertex $b_i^{H}$ of $H$ corresponds to a branching vertex $b_i^S$ in the solution of the $\MMST(S, \D, H, M, A)$ problem. For each branch $B_{u,v}$ of $H$, we compute $\Wutil_{u,v}$ and $\Wutil_{v,u}$ based on assignment $A$. This computation can be easily completed in total $O(k)$ time. Since $\Wutil_{u,v}$ and $\Wutil_{v,u}$ are complementary, the candidate region $R_{u,v}$ is uniquely defined. 
		Let $\mathcal{B}^{S}$ be the set of points in $S$ that correspond to internal vertices of $\mathcal{B}^{H}$ through mapping $M$.
		For each branch~$B_{u,v}$ with $u,v \in \mathcal{B}^{S}$, our algorithm checks whether region $R_{u,v}$ is a valid area of the plane by verifying that $v$ falls within $\Wutil_{B_{u,v}}$. If $R_{u,v}$ is valid, then $B_{u,v}$ must be contained within it; otherwise, the algorithm rejects tuple $(S, \D, H, M, A)$. The assignment $A$ of the wedges of $\W_\D$ to the leaves of $H$ defines the regions that each leaf path must be located in. We denote by $\mathcal{R}$ the set that contains all leaf path regions and branch regions determined so far. Observe that $|\mathcal{R}|=\ell+b-1=O(k)$. 
		
		Due to \cref{lem:Branches_large_Wuv}, if $P_{u,v}$ is either a branch or a leaf path of a $\D$-monotone spanning tree, then areas $\Wutil_{u,v}(u)$ and $R_{u,v}$ must be disjoint. Since $\Wutil_{u,v}(u)$ is bounded by two semi-lines originating at $u$ and $R_{u,v}$ is either a parallelogram or a strip between two parallel lines, the test for their intersection can be completed in constant time \cite{intersecting-regions}. In total, we can check in $O(k)$ time all intersections suggested by \cref{lem:Branches_large_Wuv}.
		
		We now compute, in total $O(nk)$ time, for each point~$p$ in~$S \setminus \mathcal{B}^{S}$, the region in $\mathcal{R}$ that contains~$p$.
		Due to~(I\ref{enum:path_subtree_inArea1})
		and~(I\ref{enum:path_subtree_inArea2}), every leaf
		path $P$ incident to a branching vertex $v$ in the
		solution of the $\MMST(S, \D, H, M, A)$ problem should be
		contained in $\Wutil_P(v)$ and every branch $B_{u,v}$
		between two branching vertices $u$ and $v$ should be
		contained in $R_{u,v}$. As a result, if there is a
		point $p$ that does not lie in any region in
		$\mathcal{R}$, the algorithm rejects tuple $(S, \D, H,
		M, A)$. Additionally, if, for a leaf $\lambda_j$ in $H$
		incident to vertex $b_i^{H}$, the corresponding region
		does not contain any points, then we also reject tuple
		$(S, \D, H, M, A)$ (because a missing leaf path
		induces a different HIT).
		
		The last step of the algorithm is to go through every region $R \in \mathcal{R}$ and check whether there exists a spanning path of the points in $R$ that is monotone with respect to the two directions $d_1$ and $d_2$ that are orthogonal to the boundaries of $R$. This can be achieved in $O(n \log n)$ time by sorting the points according to $d_1$ and $d_2$ and comparing whether both orderings coincide. If each region in $\mathcal{R}$ contains a $\D$-monotone path, then connecting all these paths yields a $\D$-monotone spanning tree $T$ for $S$. Observe that $T$ is unique, since in each region we have a unique $\D$-monotone path.
		
		The algorithm for solving the $\MMST(S, \D, H, M, A)$ problem terminates in $O(n\log n+nk+k)$ time. Its correctness is immediate from \cref{lem:Branches_large_Wuv}.
	\end{proof}
	
	\begin{theorem}
		\label{thm:general-k}
		Let $S$ be a set of $n$ points, and let ${\cal D}$ be a set of $k$ (pairwise non-opposite) distinct directions. There exists a function $f \colon \mathbb{N} \to \mathbb{N}$ such that, if $S$ is in ${\cal D}$-general position, then we can compute a minimum $\cal D$-monotone spanning tree of~$S$ in $O(f(k) \cdot n^{2k-1} \log n)$ time.  In other words, there is an XP algorithm for the $\MMST(S,\mathcal{D})$ problem.
	\end{theorem}
	\begin{proof}
		The given set $\mathcal{D}$ of $k$ directions yields a set of $2k$
		wedges.  Hence, a $\mathcal{D}$-monotone spanning tree has at most
		$2k$ leaves and at most $2k-2$ branching vertices.  We enumerate the
		at most $7^{2k} \cdot (2k)!$ HITs with at most $2k$ leaves
		according to \cref{le:numEmbeddedTrees}.
		Let $H$ be the current HIT, and let $\ell \le 2k$ be the number of
		leaves of~$T$. Then $T$ has at most $b=\ell-2$ branching vertices.
		We go through each of the $O(n^b)=O(n^{2k-2})$ subsets of cardinality $b \le 2k-2$ of~$S$. Let $M$ be the mapping of the internal vertices of~$H$ to points in~$S$.
		Let $A$ be the assignment of a set of consecutive wedges in $\W_\D$ to the leaves of~$H$. There are at most $2k \cdot {{2k-1} \choose {\ell-1}} \le 2k \cdot 2^{2k}$ many such assignments since we have $2k$ choices for mapping the first leaf
		to some wedge, and then we select $\ell-1$ out of the $2k-1$
		remaining wedges that we attribute to a different leaf than the
		preceding wedge (in circular order). For each of the
		$n^{2k-2}\cdot f_0(k)$, (with $f_0(k)=7^{2k} \cdot (2k)! \cdot 2k \cdot 4^k \in 2^{O(k \log k)}$) choices of a HIT~$H$, mapping $M$, and assignment $A$, we run the algorithm presented in the proof of \cref{le:mmst_sdtma} for the $\MMST(S, \D, H, M, A)$ problem, which terminates in $O(n\log n+nk+k)$ time.
		Finally, we return the shortest tree that we found (if any).
		The total runtime is $O(f(k) \cdot n^{2k-1}\log n)$, where
		$f(k)=f_0(k)\cdot k \in 2^{O(k \log k)}$.
		
		It remains to show the correctness of our approach. Toward that goal it is sufficient to show that for any $\D$-monotone spanning tree $T$ there exists a HIT $H$, a mapping $M$ and an assignment $A$, such that $T$ is the solution to the $\MMST(S, \D, H, M, A)$ problem. We fix $H$ to be the unique HIT of tree $T$ and $M$ to be the corresponding mapping of the internal vertices of $H$ to the branching vertices of $T$. We proceed to show how to specify an appropriate wedge assignment $A$.
		
		We initialize our assignment $A$ by adopting the actual wedge usage of the leaf paths of $T$.
		We now describe how to extend $A$ by assigning the remaining wedges of $\W_\D$ to the existing leaf paths of $T$ based on the branches of $T$. 
		Since $T$ is a $\D$-monotone spanning tree of $S$, for every oriented branch $B_{v,u}$ of $T$, we know $\Wutil_{u,v}$ and $\Wutil_{B_{v,u}}$. If $\Wutil_{B_{v,u}}$ contains wedges which are not included in $\Wutil_{u,v}$, then we assign  these wedges to the leading and/or the trailing leaf path that utilizes wedges in $\Wutil_{u,v}$ (refer to \cref{fig:assignment_proof}). Note that these wedges are not utilized by any other leaf path. Also note that the same wedge $W$ cannot receive contradicting assignment due to two different branches. If that were the case, then these two branches would have to be oppositely facing in the path $P$  that has them at its ends. Then, path $P$ would not be monotone since $\Wutil_{P}$ would contain both wedge $W$ and its  opposite wedge. Any remaining unassigned wedges after the processing of all branches of $T$ are assigned arbitrarily to a leaf path that utilizes the CCW neighboring wedges. The resulting assignment $A$ assigns all $2k$ wedges of $\W_\D$ to leaf paths of $H$ and is consistent with the $\D$-monotone tree $T$.
	\end{proof}

	\begin{figure}[tb]
		\centering
		\includegraphics[page=1]{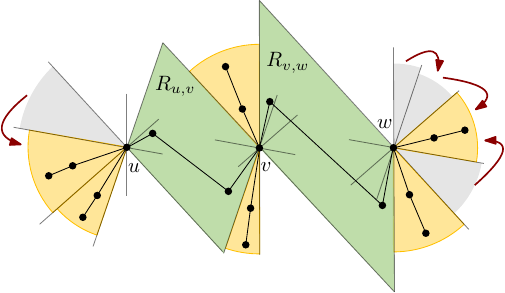}
		\caption{Gray regions are the wedges not utilized by any leaf path. The red arrows indicate the leaf to which a wedge is assigned (as described in the proof of \cref{thm:general-k}).}
		\label{fig:assignment_proof}
	\end{figure}

	Arguing as in the proof of \cref{th:2-monotone-tree}, we get the
	following result.

	\begin{theorem}
		\label{thm:directional-k}
		Given a set $S$ of $n$ point, there exists a function
		$f \colon \mathbb{N} \to \mathbb{N}$ such that we can compute a
		minimum $k$-directional spanning tree of~$S$ in
		$O(f(k) \cdot n^{2k(2k-1)} \log n)$ time.  In other words, there is
		an XP algorithm for the $\MMST(S,k)$ problem.
	\end{theorem}

	\begin{proof}
		Let $\sigma = \langle d_1, d_2, \dots, d_h \rangle$ be the circular
		sequence of directions with $h \leq {n \choose 2}$ as defined in the
		proof of \cref{th:1-monotone-tree}, and which can be computed in
		$O(n^2 \log n)$ time.  By applying \cref{thm:general-k} to every set
		of $k$ distinct directions in $\sigma$, we consider all candidate
		$\D$-monotone trees over all sets $\D$ of $k$ directions.  Since
		there are ${h \choose k} \in O(n^{2k})$ sets, this takes
		$O(f(k) \cdot n^{2k(2k-1)} \log n)$ time.
		
		It remains to argue that it suffices to restrict ourselves to
		sets~\D of $k$ directions for which $S$ is \D-monotone.  This can be
		shown by using an exchange argument as in the proof of
		\cref{th:2-monotone-tree}.
	\end{proof}
	
	\section{The Maximum Degree of the Minimum
		$k$-Directional Monotone Spanning Tree}
	\label{se:maxdeg}
	
	Since the maximum vertex degree of (Euclidean) MSTs is at most six~\cite{GEORGAKOPOULOS1987122}, it is natural to ask whether
	this upper bound carries over to minimum $k$-directional monotone
	spanning trees. We prove that this is not the case by presenting a set
	$\mathcal{D}$ of $k$ 
	specific 
	directions and a set $S_k$ of $2k+1$
	points such that the unique monotone $k$-directional
	spanning tree of~$S_k$ has degree~$2k$.
	
	Let $k$ be an even positive integer, and let
	$\mathcal{D} = \{d_1,d_2,\dots,d_k\}$ be the set of $k$ distinct
	(pairwise non-opposite) directions (in CCW order) such
	that $d_1$ is defined by the vector $(1,0)$ and, for
	$i \in [k-1]$, $\angle d_k d_{k+1}=\frac{\pi}{k}$. Since $k$ is even, it holds that
	$\mathcal{W}_{\mathcal{D}} = \mathcal{W}_{\overline{\mathcal{D}}}$ where $\overline{\D}= \{ \overline{d_1}, \overline{d_2}, \ldots, \overline{d_k} \}$. For simplicity, we consider $W_0$ to be the wedge defined by $d_1$ and $d_2$.
	We define $S_k=\{o\} \cup \{v_0,v_1, \dots, v_{2k-1}\}$ to be the set of $2k+1$ points, where $o$ is the origin and, for $i \in \{0,\ldots,2k-1\}$, $v_i$ is placed on the unit circle in the (CCW) second angle-trisection of wedge $W_i$ of $\mathcal{W}_{\mathcal{D}}$;
	see \cref{fig:pointset}. By construction, $S_k\setminus\{o\}$ is the vertex set of a regular $2k$-gon centered at~$o$ and the star with 
	edges $ov_0,\dots,ov_{2k-1}$ is a valid monotone spanning tree
	for~$S_k$ of length~$2k$. Thus, any solution of
	the $\MMST(S_k, \mathcal{D})$ problem has length at most~$2k$.
	
	Let $T$ be a tree that spans $S_k$. We call \emph{polygon vertices}
	the vertices of $T$ distinct from $o$. We refer to the edges of $T$
	connecting adjacent polygon vertices as \emph{external}, to the edges
	incident to $o$ as~\emph{rays}, and to all other edges as
	\emph{chords}.  
	
	\begin{figure}[tb]
		\centering
		\begin{subfigure}{.32\linewidth}
			\centering \includegraphics[page=1,width=\textwidth]{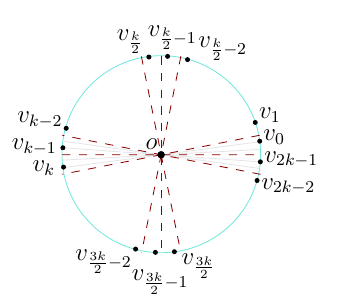}
			\subcaption{}
			\label{fig:pointset}
		\end{subfigure}
		\hfil
		\begin{subfigure}{.32\linewidth}
			\centering \includegraphics[page=8,width=\textwidth]{2k_degree_proof}
			\subcaption{}
			\label{fig:theorem_proof.1}  
		\end{subfigure}
		\hfil
		\begin{subfigure}{.32\linewidth}
			\centering \includegraphics[page=9,width=\textwidth]{2k_degree_proof}
			\subcaption{}
			\label{fig:mmsg}  
		\end{subfigure}
		\caption{(a) The point set $S_k$ is defined based on the set
			$\mathcal{W}_{\mathcal{D}}$ of wedges.  (b)~The path setting
			exploited in the proof of Theorem~\ref{theorem:deg_star}. (c) A monotone spanning graph of the point set in \cref{fig:pointset}.}
		\label{fig:pointset-full}
	\end{figure}
	
	To show that the unique solution to the $\MMST(S_k,\mathcal{D})$ problem
	 is the $2k$-star centered at~$o$, we first establish that the polygon vertices have degree at most~2.

         \begin{lemma}
           \label{lemma:DegreeTwo}
           Let $T$ be a solution to the $\MMST(S_k,\mathcal{D})$
           problem.  If~$x$ is a polygon vertex, then
           $\deg_T(x) \le 2$.
         \end{lemma}

	\begin{figure}[tb]
		\centering
		\begin{subfigure}{.45\linewidth}
			\centering
			\includegraphics[page=6]{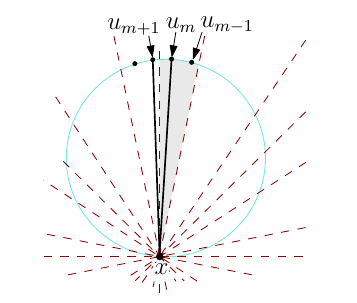}
			\subcaption{}
			\label{fig:collorary}
		\end{subfigure}
		\begin{subfigure}{.45\linewidth}
			\centering
			\includegraphics[page=7]{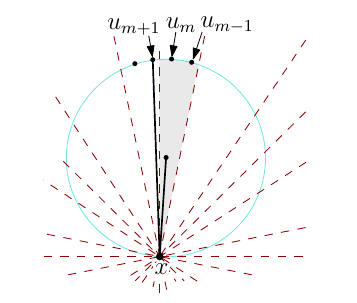}
			\subcaption{}
			\label{fig:collorary2}
		\end{subfigure}
		\caption{Connections of degree-two vertices in the proof of
			\cref{lemma:DegreeTwo}}
		\label{fig:collorary0}
	\end{figure}
	
	\begin{proof}
		Refer to \cref{fig:collorary}. Consider an arbitrary wedge of
		$\mathcal{W}_{\mathcal{D}}(x)$. The wedge contains exactly two
		(consecutive) polygon vertices, say $u_{m-1}, u_{m}$. Due to
		Property~\ref{pr:diff_wedges}, $x$ is not connected in $T$
		with both $u_{m-1}$ and $u_{m}$; otherwise, the path
		$\langle u_{m-1}, x, u_{m} \rangle$ in $T$ would not be
		$\mathcal{D}$-monotone. Given that for every three consecutive
		polygon vertices exactly two of them lie in the same wedge, they
		cannot be all three incident to $x$ in $T$. Then, by \cref{lemma:degree}, $x$ can have at most two neighbours.
		
		Consider now the case where $x$ is connected with $o$
		(\cref{fig:collorary2}). The wedge of
		$\mathcal{W}_{\mathcal{D}}(x)$ which contains $o$ also contains two polygon vertices $u_{m-1}$ and $u_{m}$ that cannot be incident to $x$ in $T$. One of them, say $u_m$, is antipodal to $x$. Then, $x$ can only be connected to the polygon vertex $u_{m+1}$ lying in the wedge adjacent to the one containing $o$ and, thus, $x$ can have degree at most two.
	\end{proof}
	
	\begin{lemma}\label{lemma:degree}
		Let $T$ be a solution to the $\MMST(S_k,\mathcal{D})$ problem and let $x \in S_k \setminus \{o\}$ be a polygon vertex having $\deg(x) \geq 2$. Then, for any two edges $(x,u_1)$ and $(x, u_2)$ that are consecutive in counter clockwise order around $x$ in $T$  and form an angle smaller than $\pi$,
		it  holds that $\angle u_1 x u_2=\frac{\pi}{2k}$. Equivalently, since the angle formed at $x$  by the edges from $x$ to  any two consecutive polygon vertices 
		is equal to $\frac{\pi}{2k}$, $u_1$ and $u_2$ are either consecutive polygon vertices, or one of them, say $u_1$, coincides with $o$ and $u_2$ is the vertex following the anti-diametric of $x$ in counterclockwise order.
	\end{lemma}
	
	\begin{proof}
		Consider the wedges in $\mathcal{W}_D (x)$ as defined by lines $\overline{d_1}(x)$, $\overline{d_2}(x)$, $\dots$, $\overline{d_k}(x)$; 
		refer to \cref{fig:polygon_vertices_1_bis}.
		Observe that the wedges partition the circle on which the polygon vertices lie into $k$ distinct circular arcs of equal length. To see this, consider an arbitrary wedge and let $y$ and $z$ be the points where it intersects the unit circle centered at $o$. Then, angle $ \angle y x z =\frac{\angle y o z }{2}$ since $x$ is on the circle with center $o$ and, thus, the circular arcs formed by the wedges are also formed by equal angles at the center of the unit polygon. As a result, the wedges at $x$ also partition the polygon vertices into $k$ distinct sets, each consisting of two vertices, since vertices are all equally spaced on the circle; see  \cref{fig:polygon_vertices_1_bis}.
		
		Assume, for the sake of contradiction, that $\angle u_1 x u_2$ is greater than $\frac{\pi}{2k}$. That is, either $u_1$ and $u_2$ are non-consecutive vertices of the polygon or one of them, say $u_1$, coincides with $o$ and $u_2$ is not the vertex counterclockwise to the anti-symmetric of $x$ with respect to $o$. We consider these two cases separately.
		
		\begin{figure}[tb]
			\centering
			\begin{subfigure}{.45\linewidth}
				\centering \includegraphics[page=2]{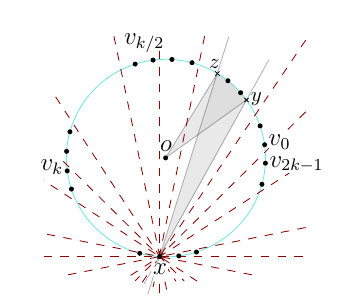}
				\subcaption{}
				\label{fig:polygon_vertices_1_bis}
			\end{subfigure}
			\begin{subfigure}{.45\linewidth}
				\centering \includegraphics[page=3]{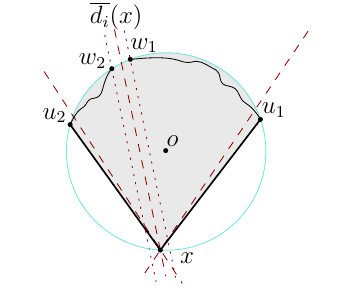}
				\subcaption{}
				\label{fig:polygon_vertices_2}
			\end{subfigure}
			\begin{subfigure}{.45\linewidth}
				\centering \includegraphics[page=5]{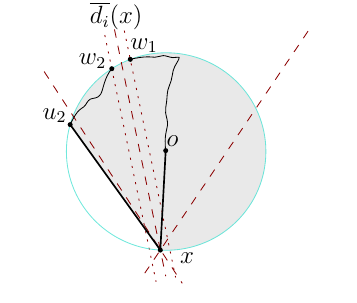}
				\subcaption{}
				\label{fig:polygon_vertices_3}
			\end{subfigure}
			\begin{subfigure}{.45\linewidth}
				\centering \includegraphics[page=4]{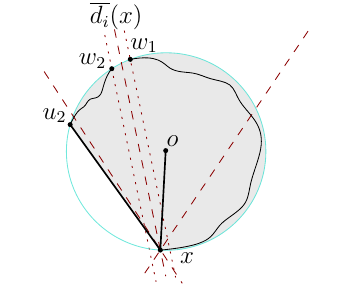}
				\subcaption{}
				\label{fig:polygon_vertices_4}
			\end{subfigure}
			\caption{(a) All but one wedges in $\mathcal{W}_{\mathcal{D}}(x)$
				contains two polygon vertices. (b-d) The path setting exploited in
				the proof of Lemma~\ref{lemma:degree}.}
			\label{fig:polygon_vertices}
		\end{figure}
		
		\begin{description}
			\item[Case 1:] {\bf $u_1$ and $u_2$ are two non-consecutive polygon vertices.} Refer to \cref{fig:polygon_vertices_2}. In this case, by rotating line $l_{x,u_1}$ counterclockwise around $x$ until it coincides with line  $l_{x, u_2}$, we define  circular sector $C$ (gray in \cref{fig:polygon_vertices_2})  which, in turn,  
			partitions $S_k$ into two subsets, namely, set $S_C$ of points lying in $C$
			and set $S_k \setminus S_C$.  All points in $S_C$ are connected to $T$ through some path either to $u_1$ or to $u_2$, since edges $(x, u_1)$ and $(x,u_2)$ are two consecutive edges around $x$ (in counterclockwise order). Now, let $w_1, w_2 \in S_C$ and $P_1=\langle x ,u_1, \dots, w_1 \rangle$ and $P_2=\langle x, u_2, \dots, w_2 \rangle$ be two paths in $T$, such that $w_1$ is the last polygon vertex connected to $u_1$ in counterclockwise order and $w_2$ is the last polygon vertex connected to $u_2$ in clockwise order. Note that one of $w_1$ or $w_2$ may coincide with $u_1$ or $u_2$, respectively. Observe that $w_1$ and $w_2$ are two consecutive vertices of the polygon. If $w_1$ and $w_2$ were not consecutive, then there would be a point $y$ between $w_1$ and $w_2$ which would not be connected to $T$.  Also, note that $w_1$ and $w_2$ should lie in different wedges of $\mathcal{W}_{\mathcal{D}}(x)$, since otherwise, due to Property~\ref{pr:diff_wedges}, path $\langle w_1, \dots, x, \dots w_2 \rangle$ would not be monotone. Let $\overline{d_i}(x)$, for some direction $d_i \in \mathcal{D}, 1\leq i \leq k$, be the line passing between $w_1$ and $w_2$. Then,  path $P=\langle w_1, \dots, x, \dots, w_2 \rangle$ must be  $d_i$-monotone. To see this, note that since path $P$ is $d$-monotone with respect to at least one direction $d \in \mathcal{D}$, its endpoints must lie on different sides of $\overline{d}(x)$. But, given that    only a single  line  through $x$ that is perpendicular to a direction in $\mathcal{D}$ can pass between  $w_1$ and $w_2$, we conclude that path $P$ is $d_i$-monotone. Thus, all of $P$'s vertices must   
			lie between the parallel lines $\overline{d_i}(w_1)$ and $\overline{d_i}(w_2)$.
			
			This is a contradiction since the strip bounded by lines $\overline{d_i}(w_1)$ and $\overline{d_i}(w_2)$ contains only $x$ and, possibly, $o$, but definitely neither $u_1$ nor~$u_2$.  Note that the case where $w_1$ coincides with $u_1$ is similar, since then all vertices of path $P$ should lie between lines $\overline{d_i}(u_1)$ and $\overline{d_i}(w_2)$. This is again a contradiction, because the strip bounded by these two lines does not contain vertex $u_2$.
			
			\item[Case 2:] {\bf One of $u_1$, $u_2$ coincides with $o$, say $u_1$, and $u_2$ is not the vertex following the antipodal point of $x$ in counterclockwise order.}  Line $l_{x,u_2}$ partitions the point set into two subsets, namely, set $S_C$ consisting of all points located on the same side of $l_{x,u_2}$  as $o$, and set $S_k \setminus S_C$. Observe that all vertices in $S_C$ are connected to $T$ via $x$, $o$ or $u_2$. We distinguish the following cases:

			\item[Case 2a:]  {\bf There is at least one vertex in $S_C$ that is connected to $T$ through $o$}. Refer to \cref{fig:polygon_vertices_3}. Again, let $w_1, w_2 \in S_C$ and $P_1=\langle x ,o, \dots, w_1 \rangle$ and $P_2=\langle x, u_2, \dots, w_2 \rangle$ be two paths in $T$ such that $w_1$ is the last polygon vertex connected to $o$ in counterclockwise order and $w_2$ is the last polygon vertex connected to $u_2$ in clockwise order. As before, $w_1$ and $w_2$ are two consecutive polygon vertices lying in different wedges of $\mathcal{W}_{\mathcal{D}}(x)$ and separated by $\overline{d_i}(x)$, for some direction $d_i \in \mathcal{D}, 1\leq i \leq k$.  Then, path $P=\langle w_1, \dots, o, x, u_2, \dots, w_2 \rangle$ must be $d_i$-monotone, which only  happens if points $o,x,u_2$ lie in the strip bounded by lines $\overline{d_i}(w_1)$ and $\overline{d_i}(w_2)$. This is a contradiction, since the only vertices in the strip are $x$ and, possibly,~$o$.
				
			\item[Case 2b:] {\bf Vertex $o$ is a leaf in $T$.} Refer to \cref{fig:polygon_vertices_4}. Define $S_C$ as in Case~2a. %
				Then, all points in $S_C$ are connected to $T$ only through $x$ and $u_2$. Let $w_1, w_2 \in S_C$ and $P_1=\langle x, \dots, w_1 \rangle$ and $P_2=\langle x, u_2, \dots, w_2 \rangle$ be two paths, such that $w_1$ is the last polygon vertex connected to $x$ in counterclockwise order and $w_2$ is the last polygon vertex connected to $u_2$ in clockwise order. Again, $w_1$ and $w_2$ are two consecutive polygon vertices lying in different wedges of $\mathcal{W}_{\mathcal{D}}(x)$ and separated by $\overline{d_i}(x)$, for some direction $d_i \in \mathcal{D}, 1\leq i \leq k$. Then, path $P=\langle w_2, u_2 \dots, x , \dots, w_1 \rangle$ is $d_i$-monotone only if $u_2$ and $x$ both lie in the strip bounded by lines $\overline{d_i}(w_1)$ and $\overline{d_i}(w_2)$; a contradiction.
				\qedhere
		\end{description}	
	\end{proof}
	
	We now show that the only solution to the $\MMST(S_k,
        \mathcal{D})$ problem is a star whose center has degree
        $2k$. Note that, by \cref{pr:max-degree-2k}, this degree bound
        is tight for a set~\D\ of $k$ directions.
	
        \begin{theorem}
		\label{theorem:deg_star}
		The only solution to  the $\MMST(S_k, \mathcal{D})$ problem is the
		star~$T^\star$ with center~$o$ and $\deg_{T^\star}(o) = 2k \in
		\Omega(|S_k|)$.
        \end{theorem}

	\begin{proof}
		Let tree $T$ be a solution to the $\MMST(S_k,\mathcal{D})$ problem
		and assume that $T$ is not the $2k$-star with $o$ at its center.  We first show that a vertex of degree one cannot be the endpoint of a chord. Let $u$ be an arbitrary polygon vertex that is a leaf of $T$ and let $(u,w)$ be the edge of $T$ that is incident to $u$. If edge $(u,w)$ was a chord, then vertices on both sides of $(u,w)$ would connect in $T$ through $w$ and, thus, $\deg(w) \geq 3$. However, this is not possible since, by  \cref{lemma:DegreeTwo}, we have that $\deg(w) \leq 2$.  Therefore, $u$ is the endpoint of either a ray or an external edge. We further observe that, as $T$ is not the $2k$-star, $T$ contains at least one external edge which has as an endpoint a polygon vertex of degree one. If this was not the case, then the polygon vertices of degree two together with $o$ would form one or more cycles, contradicting the acyclicity of $T$.
		
		\begin{figure}
			\centering
			\includegraphics[page=8]{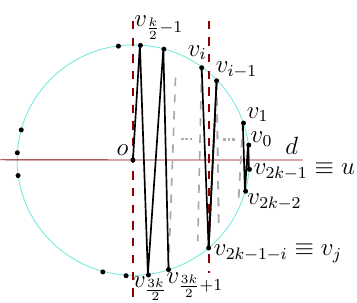}
			\caption{The path setting exploited in the proof of
				Theorem~\ref{theorem:deg_star}.}
			\label{fig:theorem_proof1}
		\end{figure}
		
		Consider now a polygon vertex $u$ of degree one that is connected to $T$ with an external polygon edge. For ease of presentation, we rotate the point set (and we renumber the vertices accordingly) so that $u$ coincides with vertex $v_{2k-1}$
		of  \cref{fig:pointset}. Note that since we have assumed that $k$
		is even, it holds that $\mathcal{W}_{\mathcal{D}} = \mathcal{W}_{\overline{\mathcal{D}}}$ and thus, one of the directions of $\mathcal{D}$, say $d$, is horizontal.
		
		We now examine how $u$, henceforth referred to as $v_{2k-1}$, is connected to $o$ in $T$. Refer to \cref{fig:theorem_proof1}. Vertex $v_{2k-1}$ cannot be connected in $T$ to $v_{2k-2}$ through the external edge $(v_{2k-1}, v_{2k-2})$. If it was, then, due to Lemma~\ref{lemma:degree}, $v_{2k-2}$ must be in turn adjacent to $v_0$. But then, both $v_0$ and $v_{2k-1}$ fall in the same wedge of $\mathcal{W}_{\mathcal{D}}(v_{2k-2})$ and, thus, the path
		$\langle v_{2k-1}, v_{2k-2}, v_{0} \rangle$ cannot be
		monotone. Therefore, $v_{2k-1}$ is connected with $v_{0}$.  Then,
		due to Lemma~\ref{lemma:degree}, $v_0$ is adjacent to $v_{2k-2}$
		which, in turn, is adjacent to $v_1$, and so on. This path
		continues until we reach the center $o$. To see this, note that if
		the path ends before reaching $o$, then it ends at a polygon vertex
		of degree one. This contradicts the fact that $T$ is a connected
		spanning tree of $S_k$. In addition, $o$ is reached through edge
		$(v_{\frac{k}{2}-1}, o)$. If this was not the case and $o$ was
		adjacent in this path to a vertex $w$ which was before
		$v_{\frac{k}{2}-1}$ in counterclockwise order, the angle formed at $w$ by $(o,w)$ and its preceding edge in the path would be greater than $\frac{\pi}{2k}$, which is impossible due to
		Lemma~\ref{lemma:degree}. Thus, $T$ contains the path $P$ that
		starts at $v_{2k-1}$, ends at $o$, and contains all points to the
		right of the vertical line through $o$. We will show that $T$ cannot be of minimum length.
		
		Consider tree $T^\prime$ formed by substituting the edges of path
		$P$ by rays from $o$ to the path vertices. Obviously $T^\prime$ is
		also monotone. To show that $T$ is not optimal, it suffices to show that the length of path $P$ is greater than the total length of the rays that replaced the edges of $P$ in $T^\prime$. In other words it suffices to show that $\lVert P \rVert > k$, where
		$\lVert P \rVert$ denotes the length of path $P$. For the length
		of path $P$, we have that $\|P\| = 1+\sum_{i=1}^{k-1}2 \sin\left(\frac{\pi}{2k}i\right)  = \cot\left(\frac{\pi}{4k}\right) > k$ (by \cref{lemma:sum,lemma:cot}).
		Thus, tree $T$ is not of minimum length; a contradiction.  We
		conclude that $T$ is the $2k$-star centered at $o$.
	\end{proof}
	
	\begin{lemma}
		\label{lemma:sum}
		$\sum_{i=1}^{k-1}2\sin(\frac{\pi}{2k}i)=\cot(\frac{\pi}{4k})-1$.
	\end{lemma}
	
	\begin{proof}
		For a sum of sine series, we know from \cite{sine-series} that:
		\begin{equation}
			\label{eq:sum}
			\sum_{i=1}^{n}\sin (a+(i-1)b)=\sin\left(a+\frac{n-1}{2}b\right)
			\frac{\sin(\frac{nb}{2})}{\sin(\frac{b}{2})}
		\end{equation}
		Additionally, from trigonometry, we know that 
		\begin{equation}
			\label{eq:trig}
			\sin(a-b)= \sin(a)\cos(b)-\cos(a)\sin(b)
		\end{equation}
		Be utilizing the above equations, we get:
		\begin{align*}
			2\sum_{i=1}^{k-1}\sin\left(\frac{\pi}{2k}i\right)  &= 2\sum_{i=1}^{k-1}\sin \left(\frac{\pi}{2k} +(i-1)\frac{\pi}{2k})\right) \\
			& \stackrel{(\ref{eq:sum})}{=} 2\left[ \sin\left( \frac{\pi}{2k}+\frac{(k-2)\pi}{4k} \right) \frac{\sin \left( \frac{(k-1)\pi}{4k} \right)}{\sin\left( \frac{\pi}{4k} \right)} \right]\\
			& \stackrel{(\ref{eq:trig})}{=} 2\left[ \sin\left(\frac{\pi}{4}\right) \frac{\sin(\frac{\pi}{4})\cos(\frac{\pi}{4k})-\cos(\frac{\pi}{4})\sin(\frac{\pi}{4k})}{\sin(\frac{\pi}{4k})} \right]\\
			&=\frac{\cos(\frac{\pi}{4k})-\sin(\frac{\pi}{4k})}{\sin(\frac{\pi}{4k})}\\
			&=\cot\left(\frac{\pi}{4k}\right)-1
		\end{align*}
	\end{proof}

	\begin{lemma}
		\label{lemma:cot}
		$\cot (\frac{\pi}{4k})>k$, for $k>1$.
	\end{lemma}
	
	\begin{proof}
	Observe that
	\begin{align*}
		\cot \left( \frac{\pi}{4k} \right)&>k \\
		\Rightarrow \frac{\cos (\frac{\pi}{4k})}{\sin (\frac{\pi}{4k})}&>k \\
		\Rightarrow \cos\left(\frac{\pi}{4k}\right)&>k\sin (\frac{\pi}{4k})\text{, since } \sin (\frac{\pi}{4k})>0 \\
		\Rightarrow \cos(u)&>\frac{\pi}{4} \cdot \frac{\sin(u)}{u} \text{, if we substitute }u=\frac{\pi}{4k}\\
	\end{align*}

	Also observe that if $u=\frac{\pi}{4k} \Rightarrow k=\frac{\pi}{4u}$
	then
	$ k \ge 2 \Rightarrow \frac{\pi}{4k} \le \frac{\pi}{8}\Rightarrow
	u\le \frac{\pi}{8}$. 
	
Due to the Cusa--Huygens inequality (see~\cite{Mitrinovic1970}),
for $0 <x < \frac{\pi}{2}$, it holds that
\begin{equation}
	\label{eq:trig2}
	\frac{\sin x}{x} < \frac{2+\cos x}{3}.
\end{equation}
	
	Now, for $u\le \frac{\pi}{8}$, from Equation \ref{eq:trig2} we
	get that
	\begin{align*}
		\frac{\pi}{4} \cdot \frac{\sin(u)}{u}< \frac{\pi}{4}\cdot\frac{2+\cos(u)}{3}
	\end{align*}
	It is now sufficient to show that 
	$$\frac{\pi}{4}\cdot\frac{2+\cos(u)}{3}<\cos(u) \text{, for all } u\le\frac{\pi}{8}$$
	$$\Rightarrow \cos(u)>\frac{2\pi}{12-\pi} \Rightarrow \cos(u)>\cos(0.78) \Rightarrow u <0.78 = 44.8^\circ$$ 
	Therefore, the inequality holds for all $u\leq \frac{\pi}{8}$ and $k\ge 2$.
\end{proof}	
	
	\section{Open Problems}\label{se:conclusions}

	We have presented an XP algorithm for solving the
	$\MMST(S,k)$ problem.  It is natural to ask  whether this problem is NP-hard if $k$ is part of the input
	(rather than a fixed constant).	
	Another research direction is to study, for a given point set~$S$ and
	a set~$\mathcal{D}$ of directions, the problem of computing a minimum
	$\mathcal{D}$-monotone spanning {\em graph} for~$S$.  Note that such a
	graph can have smaller total length than a solution to the
	$\MMST(S,\mathcal{D})$ problem.  Indeed, \cref{theorem:deg_star} shows that
	there is a point set~$S_k$ (\cref{fig:pointset}) and a set
	$\mathcal{D}$ of $k$ directions such that the only solution to
	$\MMST(S_k, \mathcal{D})$ is the $2k$-star, which has a total length
	of $2k$.  A monotone spanning {\em graph} of~$S_k$ (see
	\cref{fig:mmsg})
	consists of a regular polygon and two rays and
	has a total length of at most $2(\pi+1)$.

        \bmhead{Acknowledgements}
            
        We thank a reviewer of a previous version of this article
        for suggesting a simpler algorithm to prove
        \cref{le:monotone-path-wedge}.
        We are indebted to the reviewers of the current version of this article for their very careful reading and many helpful comments regarding notation and technical details.
        We also thank the organizers of the workshop~\emph{Graph and Network Visualization} (Heiligkreuztal 2022), where the research for this paper started. Research partially supported by: (i)~Italian MUR PRIN Proj.\ 2022TS4Y3N~-- ``EXPAND: scalable algorithms for EXPloratory Analyses of heterogeneous and dynamic Networked Data''; (ii)~Italian MUR PRIN Proj.\ 2022ME9Z78~-- ``NextGRAAL: Next-generation algorithms for constrained GRAph visuALization''
	
	\bibliography{bibliography}
	
\end{document}